\documentclass[a4paper,10pt]{article}

\usepackage{amssymb}
\usepackage{graphicx}

\usepackage{amsmath}
\usepackage{cleveref}
\usepackage{multirow}
\usepackage{tabularx}
\usepackage{xcolor}

\usepackage{url}
\urldef{\mailsc}\path|{cristina.fernandez, carlos.vela, merce.villanueva}@uab.cat|

\newtheorem{theorem}{Theorem}[section]
\newtheorem{corollary}{Corollary}[section]
\newtheorem{proposition}{Proposition}[section]
\newtheorem{lemma}{Lemma}[section]
\newtheorem{example}{Example}[section]

\newenvironment{proof}[1][Proof.]{\vspace{0.5em}\textbf{#1} }{\
\hfill$\square$}

\newcommand{\Z}{\mathbb{Z}}
\newcommand{\zero}{{\mathbf{0}}}
\newcommand{\one}{{\mathbf{1}}}

\newcommand{\C}{{\cal C}}
\newcommand{\N}{{\mathbb{N}}}

\newcommand{\wt}{{\rm wt}}
\newcommand{\rank}{\text{rank}}
\newcommand{\kernel}{\text{ker}}
\newcommand{\ord}{\operatorname{ord}}

\newcommand{\cA}{{\cal A}}

\begin{document}

\title{On the Kernel of $\Z_{2^s}$-Linear Hadamard Codes\thanks{This work has been partially supported by the Spanish MINECO under Grants TIN2016-77918-P
(AEI/FEDER, UE) and MTM2015-69138-REDT, and by the Catalan AGAUR under Grant
2014SGR-691. The authors are with the Department of Information and Communications
Engineering, Universitat Aut\`{o}noma de Barcelona, 08193 Cerdanyola del Vall\`{e}s, Spain.} \thanks{The material in this paper was presented in part at ``$5^{th}$ International Castle Meeting on Coding Theory and Applications'' in Vihula Manor, Estonia, 2017 \cite{5ICMCTA}.}}

\author{Cristina Fern\'andez-C\'ordoba, Carlos Vela, Merc\`e Villanueva}

\maketitle

\begin{abstract}
The $\Z_{2^s}$-additive codes are subgroups of $\Z^n_{2^s}$, and can
be seen as a generalization of linear codes over $\Z_2$ and $\Z_4$. A
$\Z_{2^s}$-linear
Hadamard code is a binary Hadamard code which is the Gray map image of a
$\Z_{2^s}$-additive code.
It is known that the dimension of the kernel can be used to give a complete
classification of the $\Z_4$-linear Hadamard codes. In this paper, the kernel of $\Z_{2^s}$-linear
Hadamard codes and its dimension are established for $s > 2$.
Moreover, we prove that this invariant only provides a complete classification for some values of $t$ and $s$.
The exact amount of nonequivalent such codes are given up to $t=11$ for any $s\geq 2$,
by using also the rank and, in some cases, further computations.
\end{abstract}

\section{Introduction}
\label{intro}

Let $\Z_{2^s}$ be the ring of integers modulo $2^s$ with $s\geq1$. The set of
$n$-tuples over $\Z_{2^s}$ is denoted by $\Z_{2^s}^n$. In this paper,
the elements of $\Z^n_{2^s}$ will also be called vectors over $\Z_{2^s}$ of
length $n$. A binary code of length $n$ is a nonempty subset of $\Z_2^n$,
and it is linear if it is a subspace of $\Z_{2}^n$. Equivalently, a nonempty
subset of $\Z_{2^s}^n$ is a $\Z_{2^s}$-additive if it is a subgroup of $\Z_{2^s}^n$.
Note that, when $s=1$, a $\Z_{2^s}$-additive code is a binary linear code and, when $s=2$,
it is a quaternary linear code or a linear code over $\Z_4$.

Two binary codes, $C_1$ and $C_2$, are said to be equivalent if there is a vector $\textbf{a}\in \Z_2^n$ and a
permutation of coordinates $\pi$ such that $C_2=\{ \textbf{a}+\pi(\textbf{c}) : \textbf{c} \in C_1 \}$.
Two $\Z_{2^s}$-additive codes, $\C_1$ and $\C_2$, are said to be permutation equivalent if they differ
only by a permutation of coordinates, that is, if there is a permutation of coordinates $\pi$
such that $\C_2=\{ \pi(\textbf{c}) : \textbf{c} \in \C_1 \}$.

The Hamming weight of a binary vector $\textbf{u}\in\Z_{2}^n$, denoted by $\wt_H(\textbf{u})$, is
the number of nonzero coordinates of $\textbf{u}$. The Hamming distance of two binary
vectors $\textbf{u},\textbf{v}\in\Z_{2}^n$, denoted by $d_H(\textbf{u},\textbf{v})$, is the number of
coordinates in which they differ.  Note that $d_H(\textbf{u},\textbf{v})=\wt_H(\textbf{v}-\textbf{u})$. The minimum distance of a
binary code $C$ is $d(C)=\min \{ d_H(\textbf{u},\textbf{v}) : \textbf{u},\textbf{v} \in C, \textbf{u} \not = \textbf{v}  \}$ The Lee
weight of an element $i\in\Z_{2^s}$ is $\wt_L(i)=\min\lbrace i, 2^s-i\rbrace$ and the Lee weight of a
vector $\textbf{u}=(u_1,u_2,\dots,u_n)\in\Z_{2^s}^n$ is
$\wt_L(\textbf{u})=\sum_{j=1}^n \wt_L(u_j)\in\Z_{2^s}$. The Lee distance of two
vectors $\textbf{u},\textbf{v}\in\Z_{2^s}^n$ is $d_L(\textbf{u},\textbf{v})=\wt_L(\textbf{v}-\textbf{u})$. The
minimum distance of a $\Z_{2^s}$-additive code $\C$ is $d(\C)=\min \{
d_L(\textbf{u},\textbf{v}) : \textbf{u},\textbf{v}\in\C, \textbf{u} \not = \textbf{v}  \}$.

In \cite{Sole}, a Gray map  from $\Z_4$ to $\Z_2^2$ is defined as
$\phi(0)=(0,0)$, $\phi(1)=(0,1)$, $\phi(2)=(1,1)$ and $\phi(3)=(1,0)$. There exist
different generalizations of this Gray map, which go from $\Z_{2^s}$ to
$\Z_2^{2^{s-1}}$ \cite{Carlet,dougherty,Nechaev,Krotov:2007}. The
one given in \cite{Carlet}, by Carlet, is the map
$\phi:\Z_{2^s}\rightarrow\Z_2^{2^{s-1}}$ defined as follows:
\begin{equation}\label{genGraymap}
\phi(u)=(u_{s-1},\dots,u_{s-1})+(u_0,\dots,u_{s-2})Y,
\end{equation}
where $u\in\Z_{2^s}$, $[u_0,u_1, \ldots, u_{s-1}]_2$ is the binary expansion of $u$, that is $u=\sum_{i=0}^{s-1}2^{i}u_i$ ($u_i \in \lbrace0,1\rbrace$), and $Y$ is a matrix of size $(s-1)\times2^{s-1}$ which columns are the elements of $\Z_2^{s-1}$. Note that $(u_{s-1},\dots,u_{s-1})$ and $(u_0,\dots,u_{s-2})Y$ are binary vectors of length $2^{s-1}$, and that the rows of $Y$ form a basis of a first order Reed-Muller code.
The generalization given in \cite{Nechaev} can be defined in terms of the elements of a Hadamard code \cite{Krotov:2007}. In this paper, we will
focus on Carlet's Gray map $\phi$, which is a particular case of the last one satisfying that $\sum \lambda_i \phi(2^i) =\phi(\sum \lambda_i 2^i)$ as we will see later.
Then, we define $\Phi:\Z_{2^s}^ n\rightarrow\Z_2^{n2^{s-1}}$ as the component-wise Gray map $\phi$.

Let $\C$ be a $\Z_{2^s}$-additive code of length $n$. We say that its binary image
$C=\Phi(\C)$ is a $\Z_{2^s}$-linear code of length $2^{s-1}n$.
Since $\C$ is a subgroup of
$\Z_{2^s}^n$, it is isomorphic to an abelian structure
$\Z_{2^s}^{t_1}\times\Z_{2^{s-1}}^{t_2}\times
\dots\times\Z_4^{t_{s-1}}\times\Z_2^{t_s}$, and we say that $\C$, or equivalently
$C=\Phi(\C)$, is of type $(n;t_1,\dots,t_{s})$.
Note that $|\C|=2^{st_1}2^{(s-1)t_2}\cdots2^{t_s}$.
Unlike linear codes over finite fields,
linear codes over rings do not have a basis, but there
exists a generator matrix for these codes. If $\C$ is a
$\Z_{2^s}$-additive code of type $(n;t_1,\dots,t_s)$, then a generator matrix
of $\C$ with minimum number of rows has exactly $t_1+\cdots+t_s$ rows.

Two structural properties of binary codes are the rank and
the dimension of the kernel. The rank of a binary code $C$ is simply the
dimension of the linear span, $\langle C \rangle$,  of $C$.
The kernel of a binary code $C$ is defined as
$\mathrm{K}(C)=\{\textbf{x}\in \Z_2^n : \textbf{x}+C=C \}$ \cite{BGH83}. If the all-zero vector belongs to $C$,
then $\mathrm{K}(C)$ is a linear subcode of $C$.
Note also that if $C$ is linear, then $K(C)=C=\langle C \rangle$.
We denote the rank of a binary code $C$ as $\rank(C)$ and the dimension of the kernel as $\kernel(C)$.
These parameters can be used to distinguish between nonequivalent binary codes, since equivalent ones
have the same rank and dimension of the kernel.

A binary code of length $n$, $2n$ codewords and minimum distance $n/2$ is called a Hadamard code.
Hadamard codes can be constructed from normalized Hadamard matrices \cite{Key,WMcwill}.
Note that linear Hadamard codes are in fact first order Reed-Muller codes, or
equivalently, the dual of extended Hamming codes \cite[Ch.13 \S 3]{WMcwill}.
The $\Z_{2^s}$-additive codes that, under the Gray map $\Phi$, give a
Hadamard code are called $\Z_{2^s}$-additive Hadamard codes and the
corresponding binary images are called $\Z_{2^s}$-linear
Hadamard codes.

The $\Z_4$-linear Hadamard codes of length $2^t$ can be classified by using either the rank or the dimension of the kernel \cite{Kro:2001:Z4_Had_Perf,PheRifVil:2006}.
Specifically, it is known that for a $\Z_4$-linear Hadamard code $C$ of type $(2^{t-1}; t_1, t_2)$,
$\kernel(C)=t_1+t_2+1$ if $t_1>2$, and $\kernel(C)=2t_1+t_2$ if $t_1=1$ or $2$, where $t_2=t+1-2t_1$.
For any integer $t\geq 3$ and each
$t_1 \in\{1,\ldots,\lfloor (t+1)/2 \rfloor\}$, there is a unique
(up to equivalence) $\Z_4$-linear Hadamard code of type $(2^{t-1}; t_1, t+1-2t_1)$, and all these codes are
pairwise nonequivalent, except for $t_1=1$ and $t_1=2$,
where the codes are equivalent to the linear Hadamard code \cite{Kro:2001:Z4_Had_Perf}.
Therefore, the number of nonequivalent $\Z_4$-linear Hadamard  codes of length $2^t$ is $\lfloor
\frac{t-1}{2}\rfloor$ for all $t\geq 3$, and it is 1 for $t=1$ and $t=2$.

Linear codes over $\Z_{p^s}$, which are a generalization of $\Z_{2^s}$-additive codes, were studied by Blake \cite{Blake} and Shankar \cite{Shankar} in 1975 and 1979, respectively. Nevertheless, the study of codes over rings increased significantly after the publication of some good properties of linear codes over $\Z_4$ and the definition of the Gray map \cite{Sole}. After that, $\Z_{2^s}$-additive codes and their images under the Gray map are deeply studied, for example, in \cite{Carlet}, and later in \cite{TapVeg:2003} and \cite{Gupta-paper}. In \cite{Krotov:2007}, Krotov studied $\Z_{2^s}$-linear Hadamard codes and their dual codes by using different generalizations of the Gray map. Recently, in \cite{TwoWeightSole}, considering Carlet's generalization of the Gray map, two-weight $\Z_{2^s}$-linear codes are studied. Note that $\Z_{2^s}$-linear Hadamard codes are in fact a particular case of these two-weight codes.

In this paper, in order to try to classify the $\Z_{2^s}$-linear Hadamard codes of length $2^t$, for any $t\geq 3$ and $s>2$,
we establish the kernel and its dimension for these codes. Moreover, we point out that
this invariant does not always provide a complete classification, once we fix $t\geq 3$ and $s>2$, unlike for $s=2$.
However, we give some new classification results for $t\leq 11$ and any $s>2$.
This correspondence is organized as follows.
In Section \ref{Sec:GrayMap}, we recall some results and we prove new ones related to the Carlet's generalized Gray map.
In Section \ref{Sec:construction}, we describe the construction of the $\Z_{2^s}$-linear Hadamard codes of type $(n;t_1, \dots, t_s)$ when this
Gray map is used. This result is already proved in \cite{Krotov:2007} in a more general way, but using other techniques.  
In Section \ref{Sec:Kernel}, we establish for which types these codes are linear, and
we give the kernel and its dimensions whenever they are nonlinear. In section \ref{classification},
through several examples, we show that, unlike for $s=2$, the dimension of the kernel
is not enough to classify completely $\Z_{2^s}$-linear Hadamard codes for some values of $t$ and $s$.
Moreover, we give the exact amount of nonequivalent such codes up to $t=11$ for any $s\geq 2$,
by using also the rank and, in some cases, further computations.
Finally, in Section \ref{Sec:Conclusions}, we give some conclusions and further research on this topic.


\section{Generalized Gray Map}
\label{Sec:GrayMap}

In this section, we present some general results about the Carlet's generalized Gray map, which
will be used to prove the main results related to $\Z_{2^s}$-linear Hadamard codes and given in next sections.

Let $e_i$ be the vector that has $1$ in the $i$th position and $0$ otherwise.
Let $u,v\in\Z_{2^s}$ and $[u_0,u_1,\dots,u_{s-1}]_2,[v_0,v_1,\dots,v_{s-1}]_2$ be the
binary expansions of $u$ and $v$, respectively. The operation ``$\odot$'' on
$\Z_{2^s}$ is defined as $u\odot v=\sum_{i=0}^{s-1}2^iu_iv_i$. Note that the
binary expansion of $u\odot v$ is $[u_0v_0,u_1v_1,\dots,u_{s-1}v_{s-1}]_2$.

\begin{proposition}\cite{TapVeg:2003}\label{Prop:Z2sOdot}
Let $u,v\in\Z_{2^s}$. Then, $\phi(u)+\phi(v)=\phi(u+v-2(u\odot v))$.
\end{proposition}

\begin{corollary}\label{coro:lema0}
Let $u\in\Z_{2^s}$ and $0\leq p\leq s-1$. Then, $\phi(u)+\phi(2^{p})=\phi(u+2^p-2^{p+1}u_p)$, where $[u_0,u_1 \ldots, u_{s-1}]_2$ is the binary expansion of $u$.
\end{corollary}

\begin{corollary}\label{coro:2s-1}
Let $u\in\Z_{2^s}$. Then, $\phi(u)+\phi(2^{s-1})=\phi(u+2^{s+1})$.
\end{corollary}

\begin{lemma}\label{lema3}
Let $u\in\lbrace2^{s-2},\dots,2^{s-1}-1\rbrace\cup\lbrace3\cdot2^{s-2},\dots,2^{s}-1\rbrace\subset\Z_{2^s}$. Then, $\phi(u)+\phi(2^{s-2})=\phi(u+2^{s-2}+2^{s-1})$.
\end{lemma}

\begin{proof}
By Proposition~\ref{Prop:Z2sOdot}, we have that
$\phi(u)+\phi(2^{s-2})=\phi(u+2^{s-2}-2(u\odot 2^{s-2}))$. The binary expansion
of $2^{s-2}$ is $[0,\dots,0,1,0]_2$  and, if
$u\in\lbrace2^{s-2},\dots,$ $2^{s-1}-1\rbrace\cup\lbrace3\cdot2^{s-2},\dots,2^{s}-1\rbrace$, the binary expansion of $u$ is
$[u_0, u_1, \dots, u_{s-3},$ $ 1, u_{s-1}]_2$. Then, $-2(u\odot 2^{s-2})=2^{s-1}$
and the statement follows.
 
\end{proof}

\begin{corollary}\label{corolario2}
Let $v\in\lbrace2^{s-2},3\cdot2^{s-2}\rbrace$ and $U=\lbrace2^{s-2},\dots,2^{s-1}-1\rbrace\cup\lbrace3\cdot2^{s-2},\dots,2^{s}-1\rbrace\subset\Z_{2^s}$. Then,
$$
\phi(u)+\phi(v)=\left\lbrace\begin{array}{ll}
\phi(u+v+2^{s-1}) & \textrm{if } u\in U\\
\phi(u+v)         & \textrm{if } u\in \Z_{2^s}\setminus U.\\
\end{array}\right.
$$
\end{corollary}

\begin{proof}
Straightforward from Corollary \ref{coro:lema0} and Lemma \ref{lema3}.
\end{proof}

\begin{lemma}\label{lema4}
Let $\lambda_i\in\Z_2$, $i\in\lbrace0,\dots,s-2\rbrace$. Then, $\sum_{i=0}^{s-2}\lambda_i\phi(2^ i)=\phi(\sum_{i=0}^{s-2}\lambda_i2^i)$, where $2^i\in\Z_{2^s}$.
\end{lemma}

\begin{proof}
Let $y_{i}$ be the $i$th row of $Y$. By the definition of $\phi$ given by (\ref{genGraymap}), we know that $\sum_{i=0}^{s-2}\lambda_i\phi(2^ i)=\sum_{i=0}^{s-2}\lambda_ie_{i+1}Y=\sum_{i=0}^{s-2}\lambda_iy_{i+1}=\boldsymbol{\lambda}Y$, where $\boldsymbol{\lambda}=(\lambda_0,\dots,\lambda_{s-2})$. Since $[\lambda_0,\dots,\lambda_{s-2},0]_2$ is the binary expansion of $\sum_{i=0}^{s-2}\lambda_i2^i$, then we have that
$\boldsymbol{\lambda}Y=\phi(\sum_{i=0}^{s-2}\lambda_i2^i)$.
 
\end{proof}

\begin{proposition}\cite{Carlet}\label{Prop1:Carlet}
Let $u,v\in\Z_{2^s}$. Then, $d_H(\phi(u),\phi(v))=\wt_H(\phi(u-v))$.
\end{proposition}

\begin{lemma}\label{lema2}
Let $u\in\Z_{2^s}$. Then,
$d_H(\phi(u),\phi(2^{s-1}))+d_H(\phi(u),\phi(0))=2^{s-1}$.
\end{lemma}

\begin{proof}
By the properties of the distance, we have that
$d_H(\phi(u),\phi(2^{s-1}))+d_H(\phi(u),\phi(0))=\wt_H(\phi(2^{s-1}
)-\phi(u))+\wt_H(\phi(u))$.
Then, since $\phi(2^{s-1})=\one$,
$\wt_H(\phi(2^{s-1})-\phi(u))=2^{s-1}-\wt_H(\phi(u))$, and the result follows.
 
\end{proof}

\begin{corollary}\label{corolario1}
Let $u,v\in\Z_{2^s}$. Then,
$d_H(\phi(u),\phi(v+2^{s-1}))+d_H(\phi(u),\phi(v))=2^{s-1}$.
\end{corollary}

\begin{proof}
Straightforward from Corollary \ref{coro:lema0} and Lemma \ref{lema2}.
 
\end{proof}


\section{Construction of $\Z_{2^s}$-Linear Hadamard Codes}
\label{Sec:construction}

The description of a generator matrix having minimum number of rows for a
$\Z_4$-additive Hadamard code, as long as recursive constructions of these matrices,
are given in \cite{Kro:2001:Z4_Had_Perf}.
In  \cite{Krotov:2007}, these results are generalized for any $s>2$.
In this section, we give another proof of the theorem that establishes that the constructed
matrices generate $\Z_{2^s}$-linear Hadamard codes, in the case that Carlet's Gray map is considered.

Let $T_i=\lbrace j\cdot2^{i-1}\, :\, j\in\lbrace0,1,\dots,2^{s-i+1}-1\rbrace \rbrace$ for all $i \in \{1,\ldots,s \}$.
Note that $T_1=\lbrace0,\dots,2^{s}-1\rbrace$. Let $t_1$, $t_2$,\dots,$t_s$ be nonnegative integers with $t_1\geq1$. Consider the matrix $A^{t_1,\dots,t_s}$ whose columns are of the form $\mathbf{z}^T$, $\mathbf{z}\in\lbrace1\rbrace\times T_1^{t_1-1}\times T_{2}^{t_2}\times\cdots\times T_s^{t_s}$.

\begin{example}
\label{ex:matrices}
For $s=3$, for example, we have the following matrices:
$$
\arraycolsep=1.4pt\def\arraystretch{}
A^{1,0,1}=\left(\begin{array}{cc}
1 & 1\\
0 & 4\\
\end{array}\right),\quad A^{1,1,0}=\left(\begin{array}{cccc}
1 & 1 & 1 & 1\\
0 & 2 & 4 & 6\\
\end{array}\right),\quad A^{2,0,0}=\left(\begin{array}{cccccccc}
1 & 1 & 1 & 1 & 1 & 1 & 1 & 1\\
0 & 1 & 2 & 3 & 4 & 5 & 6 & 7\\
\end{array}\right),
$$
$$
\arraycolsep=1.4pt\def\arraystretch{}
A^{1,1,1}=\left(\begin{array}{ccccccccc}
1 & 1 & 1 & 1 & & 1 & 1 & 1 & 1\\
0 & 2 & 4 & 6 & & 0 & 2 & 4 & 6\\
0 & 0 & 0 & 0 & & 4 & 4 & 4 & 4\\
\end{array}\right),\quad A^{2,0,1}=\left(\begin{array}{ccccccccccccccccc}
1 & 1 & 1 & 1 & 1 & 1 & 1 & 1 & & 1 & 1 & 1 & 1 & 1 & 1 & 1 & 1 \\
0 & 1 & 2 & 3 & 4 & 5 & 6 & 7 & & 0 & 1 & 2 & 3 & 4 & 5 & 6 & 7 \\
0 & 0 & 0 & 0 & 0 & 0 & 0 & 0 & & 4 & 4 & 4 & 4 & 4 & 4 & 4 & 4\\
\end{array}\right),
$$
$$
\arraycolsep=1.4pt\def\arraystretch{}
A^{2,1,0}=\left(\begin{array}{ccccccccccccccccccccccccccccccccccc}
1 & 1 & 1 & 1 & 1 & 1 & 1 & 1 & & 1 & 1 & 1 & 1 & 1 & 1 & 1 & 1 & & 1 & 1 & 1 & 1 & 1 & 1 & 1 & 1 & & 1 & 1 & 1 & 1 & 1 & 1 & 1 & 1\\
0 & 1 & 2 & 3 & 4 & 5 & 6 & 7 & & 0 & 1 & 2 & 3 & 4 & 5 & 6 & 7 & & 0 & 1 & 2 & 3 & 4 & 5 & 6 & 7 & & 0 & 1 & 2 & 3 & 4 & 5 & 6 & 7\\
0 & 0 & 0 & 0 & 0 & 0 & 0 & 0 & & 2 & 2 & 2 & 2 & 2 & 2 & 2 & 2 & & 4 & 4 & 4 & 4 & 4 & 4 & 4 & 4 & & 6 & 6 & 6 & 6 & 6 & 6 & 6 & 6\\
\end{array}\right).
$$
\end{example}

Let  $\mathbf{0}, \mathbf{1},\mathbf{2},\ldots, \mathbf{2^{s}-1}$ be the vectors having the elements $0, 1, 2, \ldots, 2^s-1$ from $\Z_{2^s}$ repeated in each coordinate, respectively. The order of a vector $\mathbf u$ over $\Z_{2^s}$, denoted by $\ord(\mathbf{u})$, is the smallest positive integer $m$ such that $m \mathbf{u} =\zero$.

Any matrix $A^{t_1,\dots,t_s}$ can be obtained by applying the
following iterative construction. We start with $A^{1,0,\dots,0}=(1)$. Then, if
we have a matrix $A=A^{t_1,\dots,t_s}$, for any $i\in \{1,\ldots,s\}$, we may construct the matrix
\begin{equation}\label{eq:recGenMatrix}
A_i=
\left(\begin{array}{cccc}
A & A &\cdots & A \\
0\cdot \mathbf{2^{i-1}}  & 1\cdot \mathbf{2^{i-1}} & \cdots & (2^{s-i+1}-1)\cdot \mathbf{2^{i-1}}  \\
\end{array}\right).
\end{equation}
Finally, permuting the rows of $A_i$, we obtain a matrix $A^{t'_1,\ldots,t'_s}$, where $t'_j=t_j$ for $j\not=i$ and $t'_i=t_i+1$.

\begin{example}
From the matrix $A^{1,0,0}=(1)$, we obtain the matrix $A^{2,0,0}$; and from $A^{2,0,0}$ we can construct $A^{2,0,1}$, where $A^{2,0,0}$ and $A^{2,0,1}$ are the matrices given in Example~\ref{ex:matrices}. Note that we can also generate another matrix $A^{2,0,1}$ as follows: from $A^{1,0,0}=(1)$ we obtain the matrix $A^{1,0,1}$ given in Example~\ref{ex:matrices}, and from $A^{1,0,1}$ we can construct the matrix
$$
\arraycolsep=1.4pt\def\arraystretch{}
A_1=\left(\begin{array}{cccccccccccccccc}
1 & 1 & 1 & 1 & 1 & 1 & 1 & 1 & 1 & 1 & 1 & 1 & 1 & 1 & 1 & 1 \\
0 & 4 & 0 & 4 & 0 & 4 & 0 & 4 & 0 & 4 & 0 & 4 & 0 & 4 & 0 & 4 \\
0 & 0 & 1 & 1 & 2 & 2 & 3 & 3 & 4 & 4 & 5 & 5 & 6 & 6 & 7 & 7 \\
\end{array}\right).
$$
Then, after permuting the rows of $A_1$, we have the matrix
$$
\arraycolsep=1.4pt\def\arraystretch{}
A^{2,0,1}=\left(\begin{array}{cccccccccccccccc}
1 & 1 & 1 & 1 & 1 & 1 & 1 & 1 & 1 & 1 & 1 & 1 & 1 & 1 & 1 & 1 \\
0 & 0 & 1 & 1 & 2 & 2 & 3 & 3 & 4 & 4 & 5 & 5 & 6 & 6 & 7 & 7 \\
0 & 4 & 0 & 4 & 0 & 4 & 0 & 4 & 0 & 4 & 0 & 4 & 0 & 4 & 0 & 4 \\
\end{array}\right),
$$
which is different to the matrix $A^{2,0,1}$ of Example~\ref{ex:matrices}. These two matrices $A^{2,0,1}$ generate permutation equivalent codes.
\end{example}

Along this paper, we consider that the matrices $A^{t_1,t_2,\ldots,t_s}$ are constructed recursively starting from $A^{1,0,\ldots,0}$ in the following way. First, we add $t_1-1$ rows of order $2^s$, up to obtain $A^{t_1,0,\ldots,0}$; then $t_2$ rows of order $2^{s-1}$ up to generate $A^{t_1,t_2,0,\ldots,0}$; and so on, until we add $t_s$ rows of order $2$ to achieve $A^{t_1,t_2,\ldots,t_s}$.

Let $\mathcal{H}^{t_1,\dots,t_s}$ be the $\Z_{2^s}$-additive code generated by the matrix $A^{t_1,\dots,t_s}$, where $t_1,\dots,t_s\geq0$ with $t_1\geq1$. Let $n=2^{t-s+1}$, where $t=\left(\sum_{i=1}^{s}(s-i+1)\cdot t_i\right)-1$. It is easy to see that $\mathcal{H}^{t_1,\dots,t_s}$ is of length $n$ and has $|\mathcal{H}^{t_1,\dots,t_s}|=2^{s}n=2^{t+1}$ codewords. Note that this code is of type $(n;t_1,t_2,\dots,t_s)$. Let $H^{t_1,\dots,t_s}=\Phi(\mathcal{H}^{t_1,\dots,t_s})$ be the corresponding $\Z_{2^s}$-linear code.

\begin{example}\label{remark:H1}
The code $\mathcal{H}^{1,0,\dots,0}$ is generated by $A^{1,0,\dots,0}=(1)$, so $\mathcal{H}^{1,0,\dots,0}=\Z_{2^s}$. This code has length $n=1$, cardinality $2^s$ and minimum distance $1$. Thus, $H^{1,0,\dots,0}=\Phi(\mathcal{H}^{1,0,\dots,0})$ has length $N=2^{s-1}$, cardinality $2N=2^s$ and minimum (Hamming) distance $N/2=2^{s-2}$, so it is a binary Hadamard code. Actually, $H^{1,0,\dots,0}=\Phi(\Z_{2^s})$ is the binary linear Hadamard code of length $2^{s-1}$ \cite{Carlet}, or equivalently, the first order Reed-Muller code of length $2^{s-1}$, denoted by $RM(1,s-1)$ \cite[Ch.13 \S 3]{WMcwill}.
\end{example}

The result given by Theorem \ref{Th:Z2Had} is already proved in
\cite{Krotov:2007}. In that paper,  it is shown that each $\Z_{2^s}$-linear Hadamard
code is equivalent to $H^{t_1,\ldots,t_s}$ for some $t_1,\ldots,t_s\geq0$ with
$t_1\geq 1$, considering a generalized Gray map that includes the one given by Carlet.  We present a new proof of this theorem,
in the case that Carlet's Gray map is considered.
This new proof does not use neither the dual of the $\Z_{2^s}$-additive codes nor another generalization
of the Gray map for these dual codes, unlike the proof given in \cite{Krotov:2007}.

Let $\mathcal{G}$ be a generator matrix of a $\Z_{2^s}$-additive code
$\mathcal{C}$ of length $n$. Then, $(\mathcal{G}\cdots\mathcal{G})$ is a
generator matrix of the $r$-fold replication code of $\mathcal{C}$,
$(\mathcal{C},\dots,\mathcal{C})=\lbrace(\textbf{c},\dots,\textbf{c}) :
\textbf{c} \in \mathcal{C}\rbrace$, of length $r\cdot n$.

\begin{theorem}\cite{Krotov:2007}\label{Th:Z2Had}
Let $t_1,\dots,t_s$ be nonnegative integers with $t_1\geq1$.
The $\Z_{2^s}$-linear code $H^{t_1,\dots,t_s}$ of type $(n;t_1,t_2,\dots,t_s)$
is a binary Hadamard code of length $2^t$, with
$t=\left(\sum_{i=1}^{s}(s-i+1)\cdot t_i\right)-1$ and $n=2^{t-s+1}$.
\end{theorem}

\begin{proof}
We prove this theorem by induction on the integers $t_i$, $i\in\lbrace1,\dots,s\rbrace$. First, by  Example \ref{remark:H1}, the code $H^{1,0,\dots,0}$ is a Hadamard code.

Let $\mathcal{H}=\mathcal{H}^{t_1,\dots,t_s}$ be the $\Z_{2^ s}$-additive code
of length $n$ generated by the matrix $A=A^{t_1,\dots,t_s}$. We assume that
$H=\Phi(\mathcal{H})$ is a Hadamard code of length $N=2^{s-1}n$. Let
$i\in\lbrace1,\dots,s\rbrace$. Define $A_i$ as in
(\ref{eq:recGenMatrix}) and let $\mathcal{H}_i$ be the $\Z_{2^s}$-additive
code generated by the matrix $A_i$. We have that
$\mathcal{H}_i$ is permutation equivalent to $\mathcal{H}^{t'_1,\dots,t'_s}$,
where  $t'_j=t_j$ for $j\not=i$ and $t'_i=t_i+1$. Now, we shall prove that
$H_i=\Phi(\mathcal{H}_i)$ is a Hadamard code.

Note that $\mathcal{H}_i$ can be seen as the union of $2^{s-i+1}$ cosets of the $2^{s-i+1}$-fold replication code of $\mathcal{H}$, $(\mathcal{H}, \ldots,\mathcal{H})$, which are
\begin{equation}\label{cosets}
\left(\mathcal{H},\dots,\mathcal{H} \right)+r\cdot \mathbf{w}_i,
\end{equation}
for $r\in\lbrace0,\dots2^{s-i+1}-1\rbrace$, where $\mathbf{w}_i= (0,\;\mathbf{2^{i-1}},\;2\cdot\mathbf{2^{i-1}},\;\dots,\;(2^{s-i+1}-1)\cdot\mathbf{2^{i-1}})$.

The code $\mathcal{H}$ of length $n$ has cardinality $2^sn$. It is easy to see that $\mathcal{H}_i$ has length $n_i=2^{s-i+1}n$ and cardinality $2^{2s-i+1}n$. Therefore, the length of $H_i=\Phi(\mathcal{H}_i)$ is $N_i=2^{s-1}n_i$ and the cardinality $2N_i$. Now, we just have to prove that the minimum distance of $H_i$ is $N_i/2$.

By Proposition \ref{Prop1:Carlet}, the minimum distance of $H_i$ is equal
to the minimum weight of $H_i$. Thus, we just have to check that the minimum
weight of any coset (\ref{cosets}) is $N_i/2$. When $r=0$, we have that
$\wt_H(\Phi((\mathbf{u},\dots,\mathbf{u})))=2^{s-i+1}\wt_H(\Phi(\mathbf{u}))$
$=2^{s-i+1}N/2=N_i/2$. Otherwise, when $r\not=0$, we consider
\begin{equation}\label{eq:WtD}
\wt_H(\Phi((\mathbf{u},\dots,\mathbf{u})+r\cdot\mathbf{w}_i))=d_H(\Phi((\mathbf{u},\dots,\mathbf{u})),\Phi(r\cdot\mathbf{w}_i)).\\
\end{equation}
Note that, by construction, the coordinates of any
nonnegative multiple of $\mathbf{w}_i$ can be partitioned into two multisets $V$
and $V'$ such that $|V|=|V'|=2^{s-i}$ and there is a bijection from $V$ to $V'$
mapping any element $\mathbf{v}\in V$ into an element $\mathbf{v}'\in V'$ such
that $\mathbf{v}'-\mathbf{v}=\mathbf{2^{s-1}}$. Therefore, (\ref{eq:WtD}) can
be written as
$$
\sum_{\mathbf{v}\in V}d_H(\Phi(\mathbf{u}),\Phi(\mathbf{v}))+\sum_{\mathbf{v}'\in V'}d_H(\Phi(\mathbf{u}),\Phi(\mathbf{v'}))=
$$
$$
\sum_{\mathbf{v}\in V}d_H(\Phi(\mathbf{u}),\Phi(\mathbf{v}))+d_H(\Phi(\mathbf{u}),\Phi(\mathbf{v+2^{s-1}}))=
$$
\begin{equation}\label{End}
|V|\cdot2^{s-1}n=2^{s-i}2^{s-1}n=N_i/2,
\end{equation}
where (\ref{End}) holds by Corollary \ref{corolario1}.
 
\end{proof}

\begin{example}\label{ex:Z8Hadamard}
Let $\mathcal{H}^{2,0,0}$ be the $\Z_8$-additive code generated by $A^{2,0,0}$ given in Example \ref{ex:matrices}.
The $\Z_8$-linear code $H^{2,0,0}=\Phi(\mathcal{H}^{2,0,0})$ has length $N=32$, $2N=64$ codewords and minimum (Hamming) distance $N/2=16$.
Therefore, it is a binary Hadamard code.
\end{example}


\section{Kernel of $\Z_{2^s}$-Linear Hadamard Codes}
\label{Sec:Kernel}

The computation of the kernel and its dimension for $\Z_4$-linear Hadamard codes is given in
\cite{Kro:2001:Z4_Had_Perf,PheRifVil:2006}.
In this section, we generalize these results for $\Z_{2^s}$-linear
Hadamard codes with $s>2$. First, we establish when these codes are linear, and, in the
case that they are nonlinear, we construct the kernel and compute its dimension.

\begin{proposition}\label{prop:lineal}
The $\Z_{2^s}$-linear Hadamard codes $H^{1,0,\dots,0}$ and $H^{1,0,\dots,0,1,0}$, with $s>2$, are linear.
\end{proposition}

\begin{proof}
By Example \ref{remark:H1}, we know that $H^{1,0\dots,0}$ is linear.

Now, we consider $\mathcal{H}=\mathcal{H}^{1,0,\dots,0,1,0}$ and $H=\Phi(\mathcal{H})$.
Recall that the code $\mathcal{H}$ is generated by
$$
A^{1,0,\dots,0,1,0}=\left(\begin{array}{cccc}
1 & 1 & 1 & 1\\
0 & 2^{s-2} & 2^{s-1} & 3\cdot2^{s-2}\\
\end{array}\right).
$$
Let $\pmb{\beta}_i=(2^{i},2^{i},2^{i},2^{i})$ for $0\leq i\leq s-1$,
$\pmb{\beta}_{s}=(0,2^{s-1},0,2^{s-1})$ and
$\pmb{\beta}_{s+1}=(0,2^{s-2},2^{s-1},3\cdot2^{s-2})$. Let $C$ be the linear
code generated by $B=\lbrace\Phi(\pmb{\beta}_i):0\leq i\leq s+1\rbrace$. Now, we
prove that $C\subseteq H$. Let
$\mathbf{c}=\sum_{i=0}^{s+1}\lambda_i\Phi(\pmb{\beta}_i)\in C$, where $\lambda_i\in\Z_2$. By
Corollary \ref{coro:2s-1}, we only have to see that
\begin{equation*}
\mathbf{c}'=\lambda_{s+1}\Phi(\pmb{\beta}_{s+1})+\sum_{i=0}^{s-2}\lambda_i\Phi(\pmb{\beta}_i)\in H.
\end{equation*}
On the one hand,
if $\lambda_{s+1}=0$, then we have that $\mathbf{c}'\in H$, since $\sum_{i=0}^{s-2}\lambda_i\Phi(\pmb{\beta}_i)=\Phi(\sum_{i=0}^{s-2}\lambda_i\pmb{\beta}_i)$ by Lemma \ref{lema4}. On the other hand, if $\lambda_{s+1}=1$, then we have that $\mathbf{c}'=\Phi((0,2^{s-2},2^{s-1},3\cdot2^{s-2}))+\Phi((u,u,u,u))$, where $u=\sum_{i=0}^{s-2}\lambda_i2^{i}$. Let $U=\lbrace2^{s-2},\dots,2^{s-1}-1\rbrace\cup\lbrace3\cdot2^{s-2},\dots,2^{s}-1\rbrace\subset\Z_{2^s}$. Then, by Corollary \ref{corolario2}, $\mathbf{c}'=\Phi((0,2^{s-2},2^{s-1},3\cdot2^{s-2})+(u,u,u,u)+(0,2^{s-1},0,2^{s-1}))$ if $u\in U$, and $\mathbf{c}'=\Phi((0,2^{s-2},2^{s-1},3\cdot2^{s-2})+(u,u,u,u))$ if $u\in\Z_{2^s}\setminus U$. In both cases, $\mathbf{c}'\in H$.

Since $|C|=|H|=2^{s+2}$, then $C=H$, and thus $H$ is linear.
 
\end{proof}

\medskip
Let $\textbf{u}=(u_1,\ldots,u_n)\in \Z_{2^s}^n$
and $[u_{i,0}, u_{i,1}, \ldots, u_{i,s-1}]_2$ be the binary expansion of $u_i$, $i\in \{1,\ldots,n\}$.
Let $p$ be an integer such that $p\in \{0,\ldots,s-1\}$. Then, we denote by $\textbf{u}^{(p)}$
the binary vector having in the $i$th coordinate the $p$th element
of the binary expansion of $u_i$, that is,
$\textbf{u}^{(p)}=(u_{1,p},\ldots, u_{n,p})$.

\begin{lemma}\label{lema:vectm}
If $\mathbf{v}=2^b(0,1,\dots,2^a-1) \in \Z_{2^s}^n$, with $n=2^a$ and $b\leq s-1$,
then $\wt_H(\mathbf{v}^{(p)})=2^{a-1}$ for all $p \in \{ b,\ldots, a+b-1\}$.
\end{lemma}

\begin{proof}
The $2^a$ coordinates of $\mathbf{v}$ contain exactly the $2^a$ elements of $\Z_{2^s}$ which have a binary expansion
of the form $[0,\ldots,0,v_b, v_{b+1},\ldots, v_{a+b-1},0,\ldots,0]_2$ with $v_p \in \{0,1\}$, for all $p \in \{b,\ldots, a+b-1\}$.
Note that we have $2^a$ different elements of $\Z_{2^s}$, represented by exactly $a$ binary coordinates.
Hence, half of the coordinates of $\mathbf{v}$ satisfy that $v_p=1$ and the other
half that $v_p=0$. Therefore, $\wt_H(\mathbf{v}^{(p)})=2^a/2=2^{a-1}$ for all $p \in \{ b,\ldots, a+b-1\}$.
 
\end{proof}

\medskip
As shown in \cite{Kro:2001:Z4_Had_Perf}, the codes $H^{1,t_2}$ and $H^{2,t_2}$, $t_2\geq 0$, are the only $\Z_4$-linear Hadamard codes which are linear. In \cite{Gupta-paper}, it is proved that the codes $H^{1,0,\dots,0,t_s}$, $t_s\geq0$, are linear. The next result shows that, for $s>2$ and $t_s\geq 0$, the codes $H^{1,0,\dots,0,1,t_s}$  and $H^{1,0,\dots,0,t_s}$ are linear, and they are the only $\Z_{2^s}$-linear Hadamard codes which are linear.

\begin{theorem}\label{theorem:lineal}
The codes $H^{1,0,\dots,0,1,t_s}$ and $H^{1,0,\dots,0,t_s}$, with $s>2$ and $t_s\geq0$, are the only $\Z_{2^s}$-linear Hadamard codes which are linear.
\end{theorem}

\begin{proof}
First, we show that these codes are linear by induction on $t_s$. By
Proposition \ref{prop:lineal}, the codes $H^{1,0,\dots,0}$ and $H^{1,0,\dots,0,1,0}$ are
linear. We assume that $H=\Phi(\mathcal{H})$, where
$\mathcal{H}=\mathcal{H}^{1,0,\dots,0,t_{s-1},t_s}$,
$t_{s-1}\in\lbrace0,1\rbrace$ and $t_s\geq0$, is linear. Now, we prove that
$H_s=H^{1,0,\dots,0,t_{s-1},t_s+1}$ is linear. Since $H$ is a linear
Hadamard code of length $2^{t_s+2t_{s-1}-1}$, it is the Reed-Muller code
$RM(1,t_s+2t_{s-1}-1)$ \cite[Ch.13 \S 3]{WMcwill}. By the iterative construction
(\ref{eq:recGenMatrix}), we have that
$H_s=\lbrace\Phi((\mathbf{h},\mathbf{h})+(\zero,\mathbf{v})):\mathbf{h}\in
\mathcal{H}, \mathbf{v}\in\lbrace\zero,\mathbf{2^{s-1}}\rbrace\rbrace$. By
Corollary \ref{coro:2s-1},
$H_s=\lbrace(\Phi(\mathbf{h}),\Phi(\mathbf{h})+\Phi(\mathbf{v})):\mathbf{h}\in
\mathcal{H},
\mathbf{v}\in\lbrace\zero,\mathbf{2^{s-1}}\rbrace\rbrace=\lbrace(\mathbf{h}',
\mathbf{h}'+\mathbf{v}'):\mathbf{h}'\in H,
\mathbf{v}'\in\lbrace\zero,\one\rbrace\rbrace$, which corresponds to the Reed-Muller code $RM(1,t_s+2t_{s-1}).$
Therefore, $H_s$ is linear.

Now, we prove the nonlinearity of $H=\Phi(\mathcal H)$, where $\mathcal H =\mathcal H^{1,0,\dots,0,2,0}$.
Let $\mathbf{r}=(0, 2^{s-2}, 2^{s-1}, 3\cdot2^{s-2})$.
Recall that $\mathcal H$ has length $16$ and is generated by
$$
A^{1,0,\dots,0,2,0}=\left(\begin{array}{cccccccccccccccc}
\mathbf{1} & \mathbf{1} & \mathbf{1} & \mathbf{1} \\
\mathbf{r} & \mathbf{r} & \mathbf{r} & \mathbf{r} \\
\mathbf{0} & \mathbf{2^{s-2}} & \mathbf{2^{s-1}} & \mathbf{3\cdot2^{s-2}} \\
\end{array}\right).
$$
By Corollaries \ref{coro:2s-1} and \ref{corolario2}, we have
$\Phi((\mathbf{r}, \mathbf{r}, \mathbf{r}, \mathbf{r}))+\Phi((\mathbf{0},
\mathbf{2^{s-2}}, \mathbf{2^{s-1}},$
$ \mathbf{3\cdot2^{s-2}}))=\Phi(\mathbf{z})$, where
$ \mathbf{z}=(\mathbf{r}, \mathbf{r}, \mathbf{r}, \mathbf{r})+(\mathbf{0},
\mathbf{2^{s-2}}, \mathbf{2^{s-1}},\mathbf{3\cdot2^{s-2}})+(\mathbf{0},\mathbf{u},\mathbf{0},\mathbf{u})$ and $\mathbf{u}=(0,2^{s-1},0,2^{s-1})$.
Since $\mathcal H$ is linear over  $\Z_{2^s}$, $\mathbf{z} \in \mathcal H$ if and only if $(\mathbf{0},\mathbf{u},\mathbf{0},\mathbf{u})\in\mathcal{H}$.
Since $\wt_H(\Phi((\mathbf{0},\mathbf{u},\mathbf{0},\mathbf{u})))=4\cdot 2^{s-1}=N/4$, where $N$ is the length of $H$,
$\Phi((\mathbf{0},\mathbf{u},\mathbf{0},\mathbf{u})) \not \in H$, so $\Phi(\mathbf{z}) \not \in H$.  Therefore, $H=H^{1,0,\dots,0,2,0}$ is nonlinear.

Let $H=\Phi(\mathcal{H})$, where $\mathcal{H}=\mathcal{H}^{t_1,\dots,t_s}$.
For any $i\in \{1,\ldots,s\}$, we define $H_i=\Phi(\mathcal{H}_i)$, where $\mathcal{H}_i=\mathcal{H}^{t'_1,\dots,t'_s}$, $t'_i=t_i+1$ and $t'_j=t_j$ for $j\not=i$.

Next, we consider $H=\Phi(\mathcal{H})$, where $\mathcal{H}=\mathcal{H}^{1,0,\dots,0}$, and
we prove that $H_i$ is nonlinear for any $i\in \{1,\ldots,s-2\}$.
Note that the generator matrix of $\mathcal{H}_i$ has two rows:
$\mathbf{w}_1=\one$ and  $\mathbf{w}_2=2^{i-1}(0,1,\dots,2^{s+1-i}-1)$.
By Corollary \ref{coro:lema0}, we know that $\Phi(\textbf{w}_2)+\Phi(\mathbf{2^{i-1}})=\Phi(\mathbf{w}_2+\mathbf{2^{i-1}}-2^{i}\mathbf{w}_2^{(i-1)})$.
Therefore, we just need to show that $2^{i}\mathbf{w}_2^{(i-1)}\not\in \mathcal{H}_i$. We have that $\wt_H(\textbf{w}_2^{(i-1)})=2^{s-i}$ by Lemma \ref{lema:vectm}. 
Since $2^i \not \in \{0,2^{s-1}\}$, $\wt_H(\phi(2^{i}))=2^{s-2}$. Then, $\wt_H(\Phi(2^{i}\mathbf{w}_2^{(i-1)}))=2^{s-i}\cdot 2^{s-2}= 2^{2s-2-i}$.
Recall that the length of $H$ is $N=2^t$, where $t=2s-i$. Therefore, we have that $\wt_H(\Phi(2^{i}\mathbf{w}_2^{(i-1)}))=2^{t-2}=N/4$, and then $\Phi(2^{i}\textbf{w}_2^{(i-1)})\not\in H_i$.


Finally, in general, for $H=\Phi(\mathcal{H})$, where $\mathcal{H}=\mathcal{H}^{t_1,\dots,t_s}$,
we prove that if $H$ is nonlinear, then $H_i$ is nonlinear for any $i\in\{1,\ldots,s\}$. Assume that $H_i$ is linear.
Then, by the iterative construction (\ref{eq:recGenMatrix}), for any $\mathbf{u},\mathbf{v}\in \mathcal{H}$, we have that
 $(\mathbf{u},\dots,\mathbf{u}),(\mathbf{v},\dots,\mathbf{v})\in \mathcal{H}_i$.
 Moreover, since $H_i$ is linear, $\Phi((\mathbf{u},\dots,\mathbf{u}))+\Phi((\mathbf{v},\dots,\mathbf{v}))=\Phi(  (\mathbf{a},\dots,\mathbf{a}) +\lambda \cdot 2^{i-1} (0, 1, \dots, 2^{s-i+1}-1))\in H_i$, where $\mathbf{a}\in \mathcal{H}$ and $\lambda \in \Z_{2^s}$.
Therefore, $\Phi(\mathbf{u})+\Phi(\mathbf{v})=\Phi(\mathbf{a}) \in H$,
and we have that $H$ is linear and the result follows.
 
\end{proof}

\medskip
Let $A^{t_1,\dots,t_s}$ be the generator matrix of $\mathcal{H}^{t_1,\dots,t_s}$, considered along this paper,
and let $\mathbf{w}_i$ be the $i$th row vector of $A^{t_1,\dots,t_s}$.
By construction, $\mathbf{w}_1=\one$ and $\ord(\mathbf{w}_i)\leq \ord(\mathbf{w}_j)$ if $i>j$. We define $\sigma\in\lbrace1,\dots,s\rbrace$ as the integer such that $\ord(\textbf{w}_2)=2^{s+1-\sigma}$. Note that $\sigma=1$ if $t_1>1$, and $\sigma=\min\lbrace i : t_i>0, i\in\lbrace2,\dots,s\rbrace\rbrace$ if $t_1=1$. In the case $\sigma=s$, the code is $\mathcal{H}^{1,0,\dots,0,t_s}$, which is linear.

\begin{example}
Considering all nonnegative integer solutions with $t_1\geq 1$ of the equation
$5=3t_1+2t_2+t_3-1$, we have that
the $\Z_8$-linear Hadamard codes of length $2^t=32$ are the following: $H^{1,0,3}$, $H^{1,1,1}$ and $H^{2,0,0}$. By Theorem \ref{theorem:lineal}, we have that $H^{1,0,3}$ and $H^{1,1,1}$ are linear, so $\kernel(H^{1,0,3})=\kernel(H^{1,1,1})=6$. By the same theorem, we also have that $H^{2,0,0}$ is nonlinear, so $\kernel(H^{2,0,0})<6$.
\end{example}

\begin{proposition}\label{Prop:kernel1}
Let $\mathcal{H}=\mathcal{H}^{t_1,\dots,t_s}$ be the $\Z_{2^s}$-additive Hadamard code of type $(n; t_1,$ $\dots, t_s)$ such that $\Phi(\mathcal{H})$ is nonlinear. Let $\mathcal{H}_b$ be the subcode of $\mathcal{H}$ which contains all the codewords of order two. Let $P=\lbrace\mathbf{2^p}\rbrace_{p=0}^{\sigma-2}$ if $\sigma\geq2$, and $P=\emptyset$ if $\sigma=1$. Then,
$$ \left\langle\Phi(\mathcal{H}_b),\Phi(P), \Phi(\sum_{i=0}^{s-2}\mathbf{2^i})\right\rangle\subseteq K(\Phi(\mathcal{H}))$$
and $\kernel(\Phi(\mathcal{H}))\geq\sigma+\sum_{i=1}^{s}t_i$.
\end{proposition}

\begin{proof}
Let $H=\Phi(\mathcal{H})$ and $\tau=\sum_{i=1}^{s}t_i$. Let $Q=\lbrace (\ord(\textbf{w}_q)/2)\mathbf{w}_q \rbrace_{q=0}^{\tau}$. Since
$\mathcal{H}_b$ contains all the elements of $\mathcal{H}$ of order two, we
have that the set $\Phi(Q)$ is a base for the
binary linear subcode $H_b=\Phi(\mathcal{H}_b)$ of $H$. By
Corollary \ref{coro:2s-1}, for all
$\mathbf{b}\in\mathcal{H}_b$ and
$\mathbf{u}\in\mathcal{H}$, we have that
$\Phi(\mathbf{b})+\Phi(\mathbf{u})=\Phi(\mathbf{b}+\mathbf{u})\in H$ and,
therefore, $H_b\subseteq K(H)$.

Assume $\sigma\geq2$. Now, we prove that $\Phi(\mathbf{2^p})\in K(H)$ for all $p \in \{0,\ldots, \sigma-2\}$. Equivalently, we show that $\Phi(\mathbf{2^p})+\Phi(\textbf{u})\in H$ for all $\textbf{u}\in\mathcal{H}$. If $\textbf{u}\in\mathcal{H}$, then $\textbf{u}=\lambda \cdot \one+\textbf{u}'$, where $\lambda\in\Z_{2^s}$ and $\ord(\textbf{u}')\leq \ord(\textbf{w}_2)=2^{s+1-\sigma}$. Let $\textbf{u}=(u_1,\ldots,u_n)\in \Z_{2^s}^n$
and $[u_{i,0}, u_{i,1}, \ldots, u_{i,s-1}]_2$ be the binary expansion of $u_i$, $i\in \{1,\ldots,n\}$. Let
$[\lambda_0,\lambda_1,\ldots,\lambda_{s-1}]_2$ be the binary expansion of $\lambda\in \Z_{2^s}$.
By Corollary \ref{coro:lema0}, we have that $\Phi(\mathbf{2^p})+\Phi(\mathbf{u})=\Phi(\mathbf{2^p}+\mathbf{u}-2^{p+1}\textbf{u}^{(p)})$, where
$\textbf{u}^{(p)}=(u_{1,p},\ldots, u_{n,p})$. Note that if $v\in\Z_{2^s}$ is of order $2^j$, then its binary expansion is of the form $[0,\dots,0,1,v_{s-j+1},\dots,v_{s-1}]_2$.
Since $p\in \{0,\ldots,\sigma-2\}$ and $\ord(\textbf{u}')\leq2^{s+1-\sigma}$, we have that
$\textbf{u}^{(p)}=(\lambda_p,\ldots,\lambda_p)$.
Therefore, $2^{p+1}\mathbf{u}^{(p)}=\lambda_p\mathbf{2^{p+1}}\in \mathcal{H}$ and $\Phi(\mathbf{2^p})+\Phi(\textbf{u})=\Phi(\mathbf{2^p}+\textbf{u}-\lambda_p\mathbf{2^{p+1}})\in H$.

Next, we show that $\Phi(\sum_{i=0}^{s-2}\mathbf{2^i})\in K(H)$. Let $\textbf{u}=(u_1,\dots,u_{n})\in\mathcal{H}$ and $\textbf{v}=(v_1,\dots,v_n)=\sum_{i=0}^{s-2}\mathbf{2}^i$. First, we prove that $\phi(v_i)+\phi(u_i)=\phi(v_i+u_i-2u_i)$ for all $i\in\lbrace1,\dots,n\rbrace$. Note that the binary expansion of $v_i$ and $u_i$ are $[1,\dots,1,0]_2$ and $[u_{i,0},u_{i,1},\ldots, u_{i,{s-1}}]_2$, respectively. Then, it is easy to check that $2(v_i\odot u_i)=2u_i$. Therefore, by  Proposition \ref{Prop:Z2sOdot}, $\phi(v_i)+\phi(u_i)=\phi(v_i+u_i-2u_i)$. Hence, $\Phi(\textbf{v})+\Phi(\textbf{u})=\Phi(\textbf{v}+\textbf{u}-2\textbf{u})\in H$ for all $\textbf{u}\in \mathcal{H}$.


Finally, we have to see that the elements of the set $\lbrace \Phi(Q), \Phi(P), \Phi(\sum_{i=0}^{s-2}\mathbf{2^i})\rbrace$ are linearly independent. By construction, the generator matrix $A^{t_1,\ldots,t_s}$ is a block upper triangular matrix, so it is easy to see that the codewords in $\Phi(Q)$ are linearly independent of the ones in $\lbrace\Phi(P), \Phi(\sum_{i=0}^{s-2}\mathbf{2^i})\rbrace$. Note that $\sigma <s$ since $H$ is nonlinear.
Thus, by Lemma \ref{lema4}, it is easy to see that the codewords in $\lbrace\Phi(P), \Phi(\sum_{i=0}^{s-2}\mathbf{2^i})\rbrace$ are linearly independent. Therefore, we have that the dimension of the linear span of this set is $\sigma+\tau$, so $\kernel(H)\geq\sigma+\tau$.
 
\end{proof}

\begin{lemma}\label{lema:sumodot}
Let $v\in\Z_{2^s}$ and $\lambda_i\in\Z_2$, $i\in\lbrace0,\dots,s-1\rbrace$. Then, $v\odot\sum_{i=0}^{s-1}\lambda_i2^i=\sum_{i=0}^{s-1}v\odot \lambda_i 2^i$.
\end{lemma}

\begin{proof}
Let $v\in\Z_{2^s}$ and $[v_0,v_1,\dots,v_{s-1}]_2$ its binary expansion. By definition, we have that $v\odot\sum_{i=0}^{s-1}\lambda_i2^i=\sum_{i=0}^{s-1}v_i\lambda_i2^i$. Note that $v_i\lambda_i2^i=v\odot\lambda_i2^i$, so $v\odot\sum_{i=0}^{s-1}\lambda_i2^i=\sum_{i=0}^{s-1}v\odot\lambda_i 2^i$.
 
\end{proof}

\begin{lemma}\label{lema:homoker}
	Let $\mathcal{H}=\mathcal{H}^{t_1,\dots,t_s}$ be the $\Z_{2^s}$-additive Hadamard code of type $(n; t_1,$ $\dots, t_s)$.
	Let $\mathcal{N}=\{ \sum_{i=\sigma-1}^{s-2}\lambda_i\mathbf{2^i} : \lambda_i\in\Z_2\}\setminus\{\sum_{i=\sigma-1}^{s-2}\mathbf{2^i}\}$ if $\sigma\leq s-1$.
	Then, $\Phi(\mathcal{N})\cap K(\Phi(\mathcal{H}))=\{\zero\}$.
\end{lemma}

\begin{proof}
    Let $H=\Phi(\mathcal{H})$. Let
$\mathbf{u}=\sum_{i=\sigma-1}^{s-2}\lambda_i\mathbf{2^i}\in\mathcal{N}$ such
that $\Phi(\mathbf{u})\in K(H)$. We want to prove that $\mathbf{u}=\mathbf{0}$.

By construction, the second row
$\mathbf{w}_2$ of $A^{t_1,\dots,t_s}$ is a $2^{t-2s+\sigma}$-fold replication
of
$\mathbf{v}=2^{\sigma-1}(0,1,\dots,2^{s+1-\sigma}-1)$, and
$\ord(\mathbf{w}_2)=2^{s+1-\sigma}$. By Proposition \ref{Prop:Z2sOdot}, we
have that $\Phi(\textbf{w}_2)+\Phi(\mathbf{u})=\Phi(\textbf{w}_2+\mathbf{u}-
2(\textbf{w}_2\odot \mathbf{u}))$. Since $\Phi(\mathbf{u})\in K(H)$,
$2(\textbf{w}_2\odot \mathbf{u})\in\mathcal{H}$.  Note that, by
Lemma \ref{lema:sumodot}, we have that $2(\textbf{w}_2\odot
\mathbf{u})=2\sum_{i=\sigma-1}^{s-2}\mathbf{w}_2 \odot \lambda_i \mathbf{2^i}
=2\sum_{i=\sigma-1}^{s-2}\lambda_i\mathbf{w}_2^{(i)}2^i\in\mathcal{H}$.

Let $\tau=\sum_{i=1}^s t_i$. If $\tau=2$, then $\mathcal{H}$ has length $2^{s+1-\sigma}$ and the only vectors in $A^{t_1,\dots,t_s}$ are $\mathbf{1}$ and $\mathbf{w}_2=\mathbf{v}$. If $\tau\geq 3$, for $i\in\lbrace3,\dots,\tau\rbrace$, the $i$th row $\mathbf{w}_i$
of $A^{t_1,\dots,t_s}$ contains zeros in the first $2^{s+1-\sigma}$ coordinates by construction. Since $\sigma \leq s-1$, $\tau \geq 2$, and hence any element of $\mathcal{H}$ restricted to the first  $2^{s+1-\sigma}$ coordinates is of the form $\mu_1\one+\mu_2\textbf{v}$ for some $\mu_1,\mu_2\in\Z_{2^s}$. We have that $2\sum_{i=\sigma-1}^{s-2}\lambda_i\mathbf{w}_2^{(i)}2^i$ restricted to the first $2^{s+1-\sigma}$ coordinates is $2\sum_{i=\sigma-1}^{s-2}\lambda_i\mathbf{v}^{(i)}2^i$, so we have to find
$\mu_1,\mu_2\in\Z_{2^s}$ such that $2\sum_{i=\sigma-1}^{s-2}\lambda_i\mathbf{v}^{(i)}2^i=\mu_1\one+\mu_2\mathbf{v}$.

    Since the first coordinate of $\mathbf{v}$ is 0, the first coordinate of
$\mathbf{v}^{(i)}$ is $0$ for all $i$.
Then, we have that $\mu_1=0$, so
$2\sum_{i=\sigma-1}^{s-2}\lambda_i\mathbf{v}^{(i)}2^i=\mu_2\mathbf{v}$. Note
that
$\mathbf{v}=\sum_{i=0}^{s-1}\mathbf{v}^{(i)}2^i=\sum_{i=\sigma-1}^{s-1}\mathbf{v
}^{(i)}2^i$. Therefore,
$2\sum_{i=\sigma-1}^{s-2}\lambda_i\mathbf{v}^{(i)}2^i=\mu_2\sum_{i=\sigma-1}^{
s-1}\mathbf{v}^{(i)}2^i$.
Since $\mathbf{u} \in \mathcal{N}$, there exists $j\in\{\sigma-1,\dots,s-2\}$
such that $\lambda_j=0$. Then, regrouping the terms, we obtain that
    $$
\sum_{\substack{i=\sigma-1\\i\not=j}}^{s-2}(\mu_2-2\lambda_i)\mathbf{v}^{(i)}
2^i+\mu_2\mathbf{v}^{(j)}2^j+\mu_2\mathbf{v}^{(s-1)}2^{s-1}=\zero.
    $$
Note that $\{\mathbf{v}^{(i)}\}_{i=\sigma-1}^{s-1}$ is a
subset of a
basis of the $RM(1,t)$. Then, we have that $(\mu_2-2\lambda_i)2^i=0$, for
$i\in\{\sigma-1,\cdots,s-2\}\setminus\{j\}$, $\mu_22^j=0$ and $\mu_22^{s-1}=0$.
As a result, $\mu_2=0$ and $\lambda_i=0$ for all $i\in\{\sigma-1,\cdots,s-2\}$.
Hence, $\mathbf{u}=\sum_{i=\sigma-1}^{s-2}\lambda_i\mathbf{2^i}=\mathbf{0}$, and
the result holds.
 
\end{proof}

\begin{lemma}\label{lema:nohomoker}
	Let $\mathcal{H}=\mathcal{H}^{t_1,\dots,t_s}$ be the $\Z_{2^s}$-additive Hadamard code of type $(n; t_1,$ $\dots, t_s)$.
	Let $\mathbf{w}_i$ be the $i$th row of $A^{t_1,\dots,t_s}$ and $\tau=\sum_{i=1}^st_i$. Let $\mathcal{M}=\lbrace \mathbf{v}=\sum_{i=2}^{\tau-t_s}\lambda_i\mathbf{w}_i:\lambda_i\in\Z_{2^s},\,\ord(\mathbf{v})>2\rbrace$, $\mathcal{N}=\{ \sum_{i=\sigma-1}^{s-2}\lambda_i\mathbf{2^i} : \lambda_i\in\Z_2 \}\setminus\{\sum_{i=\sigma-1}^{s-2}\mathbf{2^i}\}$ if $\sigma\leq s-1$ and $\mathcal{M}+\mathcal{N}=\{\mathbf{v}_\mathcal{M}+\mathbf{v}_\mathcal{N}:\mathbf{v}_\mathcal{M}\in\mathcal{M}\cup\{\zero\},\mathbf{v}_\mathcal{N}\in\mathcal{N}\}$. Then, $\Phi( \mathcal{M+N})\cap K(\Phi(\mathcal{H}))=\{\zero\}$.
\end{lemma}

\begin{proof}
Let $H=\Phi(\mathcal{H})$, which has length $N=2^t=n\cdot 2^{s-1}$. By Lemma \ref{lema:homoker}, we already know that $\Phi(\mathcal{N})\cap K(H)=\{\mathbf{0}\}$.
Now, we prove that $\Phi(\mathcal{M})\cap K(H)=\emptyset$.

Let $\textbf{v}=\sum_{i=2}^{\tau-t_s}\lambda_i\textbf{w}_i \in \mathcal{M}$.
Since $\ord(\mathbf{v})>2$ and $\ord(\mathbf{w}_i)\leq 2^{s+1-\sigma}$, $\ord(\textbf{v})=2^{p}$ for some $2\leq p\leq s+1-\sigma$. By the iterative construction (\ref{eq:recGenMatrix}) of $A^{t_1,\dots,t_s}$, we know that
all the elements of $\Z_{2^s}$ of order equal to or less than $2^{p}$ appear as a coordinate of $\mathbf{v}$.
Moreover, exactly half of the coordinates of $\textbf{v}$ are of order $2^{p}$.
We consider two cases depending on the value of $p$.

First, we consider that $2<p\leq s+1-\sigma$. We have that $\Phi(\textbf{v})+\Phi(\mathbf{2^{s-p}})=\Phi(\textbf{v}+\mathbf{2^{s-p}}-2^{s-p+1}\textbf{v}^{(s-p)})$ by Corollary \ref{coro:lema0}. As before, it is enough to see that $2^{s-p+1}\textbf{v}^{(s-p)} \not\in\mathcal{H}$ to prove that $\Phi(\textbf{v}) \not \in K(H)$. Since half of the coordinates of $\textbf{v}$ are of order $2^{p}$ and the other half are of order less than $2^p$, we have that half of the coordinates of $2^{s-p+1}\textbf{v}^{(s-p)}$ are equal to $2^{s-p+1}$ and the rest of coordinates are zero.
Note that $2^{s-p+1} \not\in \{0,2^{s-1}\}$ since $p>2$. Therefore, since $\wt_H(\phi(2^{s-p+1}))=2^{s-2}$,
we have that $\wt_H(\Phi( 2^{s-p+1}\textbf{v}^{(s-p)}))= n/2 \cdot 2^{s-2}=2^{t-2}=N/4$ and hence $\Phi(\mathbf{v}) \not \in K(H)$.

Next, we consider that $p=2$, that is, $\ord(\mathbf{v})=4$. Then, $\ord(\lambda_i\mathbf{w}_i)=4$ or $\lambda_i=0$ for all $i\in\{2,\dots,\tau-t_s\}$.
By Proposition \ref{Prop:Z2sOdot},  $\Phi(\mathbf{v})+\Phi(2^{s-\sigma-1}\mathbf{w}_2)=\Phi(\mathbf{v}+2^{s-\sigma-1}\mathbf{w}_2-2(\mathbf{v}\odot 2^{s-\sigma-1}\mathbf{w}_2))$. Again, it is enough to see that $2(\mathbf{v}\odot2^{s-\sigma-1}\mathbf{w}_2)\not\in\mathcal{H}$ to show that $\Phi(\mathbf{v})\not\in K(H)$. Note that $2^{s-\sigma-1}\mathbf{w}_2$ is a $2^{t-s-1}$-fold replication of $\textbf{b}_1=(0,2^{s-2},2^{s-1},3\cdot2^{s-2})$.
Now, we consider the coordinates divided into groups of 4 consecutive coordinates, which will be referred to as blocks.
Note that every block of $\lambda_i\textbf{w}_i$ contains the same value in its 4 coordinates,
for all $i\in\lbrace3,\dots,\tau-t_s\rbrace$.

If $\lambda_2=0$, then every block of $\textbf{v}$ also contains the same value in its 4 coordinates.
Thus, every block in $2(\mathbf{v}\odot2^{s-\sigma-1}\mathbf{w}_2)$ is of the form
$2(\mathbf{k}\odot \textbf{b}_1)$ for some $k\in\{0,2^{s-2},2^{s-1},3\cdot2^{s-2}\}$. We have that
\begin{equation*}
2(\mathbf{k}\odot \textbf{b}_1)=\left\lbrace\begin{array}{cl}
(0,0,0,0) & \text{if }k\in\{0,2^{s-1}\}\\
(0,2^{s-1},0,2^{s-1}) & \text{if }k\in\{2^{s-2},3\cdot2^{s-2}\}.\\
\end{array}\right.
\end{equation*}
By construction, note that $\mathbf{v}$ contains the same number of blocks $\mathbf{k}$ for each $k\in \{0,2^{s-2},2^{s-1},3\cdot2^{s-2}\}$.
Then, it is easy to see that $\wt_H(\Phi(2(\mathbf{v}\odot2^{s-\sigma-1}\mathbf{w}_2)))=\wt_H(\phi(2^{s-1}))\cdot 4 \cdot n/16=2^{s-1}\cdot n/4=2^{t-2}=N/4$, so $\Phi(\mathbf{v})\not\in K(H)$ in this case.

Otherwise, if $\lambda_2\not=0$, then every block of $\mathbf{v}$ is of the form $\textbf{b}_i+\mathbf{k}$, for some $i\in\{1,2\}$ and $k\in\{0,2^{s-2},2^{s-1},3\cdot2^{s-2}\}$, where $\textbf{b}_1=(0,2^{s-2},2^{s-1},3\cdot2^{s-2})$ and $\textbf{b}_2=(0,3\cdot2^{s-2},2^{s-1},2^{s-2})$.
Then, we have that
\begin{equation*}
2((\textbf{b}_i+\mathbf{k})\odot \textbf{b}_1)=\left\lbrace\begin{array}{cl}
(0,0,0,0) & \text{if }k\in\{2^{s-2},3\cdot2^{s-2}\}\\
(0,2^{s-1},0,2^{s-1}) & \text{if }k\in\{0,2^{s-1}\},\\
\end{array}\right.
\end{equation*}
for $i\in\{1,2\}$. Again, by construction, $\mathbf{v}$ contains the same number of blocks $\textbf{b}_i+\mathbf{k}$ for each $k\in \{0,2^{s-2},2^{s-1},3\cdot2^{s-2}\}$.
Therefore, as before, $\wt_H(\Phi(2(\mathbf{v}\odot2^{s-\sigma-1}\mathbf{w}_2)))=N/4$, and $\Phi(\mathbf{v})\not\in K(H)$.
We have just shown that $\Phi(\mathcal{M})\cap K(H)=\emptyset$.

Now, we prove that $\Phi(\mathcal{M+N})\cap K(H)=\{\zero\}$.
Let $\mathbf{v}=\mathbf{v}_\mathcal{M}+\mathbf{v}_\mathcal{N} \in\mathcal{M+N}\backslash \{\zero\} $,
where $\mathbf{v}_\mathcal{M}\in\mathcal{M}$ and $\mathbf{v}_\mathcal{N}\in\mathcal{N}$. We just proved that $\Phi(\mathbf{v}) \not \in K(H)$ if $\mathbf{v}_\mathcal{M}=\zero$ or $\mathbf{v}_\mathcal{N}=\zero$. Therefore, we can assume that $\mathbf{v}_\mathcal{M}\not =\zero$ and $\mathbf{v}_\mathcal{N}\not =\zero$.

We know that $\mathbf{v}_\mathcal{N}=(v,\dots,v)$. Let $[v_0,v_1,\dots,v_{s-1}]_2$ be the binary expansion of $v$.
Let $v_{\mathcal{N}_1}$ and $v_{\mathcal{N}_2}$ be the elements of $\Z_{2^s}$ having binary expansion $[0,\dots,0,v_{s-p},\dots,v_{s-1}]_2$ and $[v_{0},\dots,v_{s-p-1},0,\dots,0]_2$, respectively. Then, $\mathbf{v}_\mathcal{N}=\mathbf{v}_{\mathcal{N}_1}+\mathbf{v}_{\mathcal{N}_2}$, where
$ \mathbf{v}_{\mathcal{N}_i} = (v_{\mathcal{N}_i}, \ldots, v_{\mathcal{N}_i})$ for $i \in \{1,2\}$.
Since $\ord(\mathbf{v}_\mathcal{M})=2^p$  with $2\leq p\leq s+1-\sigma$, the binary expansion of each one of its coordinates of is of the form $[0,\dots,0,(v_{\mathcal{M}})_{s-p},\dots,(v_{\mathcal{M}})_{s-1}]_2$.
Note that we also have that $\ord(\mathbf{v}_{\mathcal{N}_1}) \leq \ord(\mathbf{v}_\mathcal{M})$ by construction.

On the one hand, we consider $2< p\leq s+1-\sigma$. It is easy to see that $2(\mathbf{v}_{\mathcal{N}_2}\odot\mathbf{2^{s-p}})=\zero$.  Therefore, $\wt_H(\Phi(2(\mathbf{v}\odot\mathbf{2^{s-p}})))=\wt_H(\Phi(2((\mathbf{v}_{\mathcal M}+\mathbf{v}_{\mathcal{N}_1})\odot\mathbf{2^{s-p}})))$.
Since $\ord(\mathbf{v}_{\mathcal{N}_1}) \leq \ord(\mathbf{v}_\mathcal{M})$, it is easy to see that there exists a permutation of coordinates $\pi$ such that
$\pi(\mathbf{v}_{\mathcal M}+\mathbf{v}_{\mathcal{N}_1} )= \mathbf{v}_{\mathcal M}$. Thus, $ \wt_H(\Phi(2((\mathbf{v}_{\mathcal M}+\mathbf{v}_{\mathcal{N}_1})\odot\mathbf{2^{s-p}}))) =\wt_H(\Phi(2(\mathbf{v}_{\mathcal M}\odot\mathbf{2^{s-p}})))$ and
the result holds by using the same arguments as above.

On the other hand, we consider that $p=2$. Note that $\ord(\mathbf{v}_\mathcal{M})=4$, and then $\ord(\mathbf{v}_{\mathcal{N}_1})=4$. It is easy to see that $2(\mathbf{v}_{\mathcal{N}_2}\odot 2^{s-\sigma-1}\mathbf{w}_2)=\zero$, hence we have that  $\wt_H(\Phi(2(\mathbf{v}\odot 2^{s-\sigma-1}\mathbf{\mathbf{w}_2})))=\wt_H(\Phi(2((\mathbf{v}_{\mathcal M}+\mathbf{v}_{\mathcal{N}_1})\odot 2^{s-\sigma-1} \mathbf{ \mathbf{w}_2})))$. Recall that $2^{s-\sigma-1}\mathbf{w}_2$ is the $2^{t-s-1}$-fold replication of $\mathbf{b}_1$. Taking into account that $\mathbf{v}_\mathcal{M}=\sum_{i=2}^{\tau-t_s}\lambda_i\textbf{w}_i$, note that the blocks of $\mathbf{v}_{\mathcal M}+\mathbf{v}_{\mathcal{N}_1}$ are of the form $\mathbf{k}$ for some $k\in\{0,2^{s-2},2^{s-1},3\cdot2^{s-2}\}$ if $\lambda_2=0$; or $\mathbf{b}_i+\mathbf{k}$ for some $k\in\{0,2^{s-2},2^{s-1},3\cdot2^{s-2}\}$ and $i\in \{1,2\}$ if $\lambda_2\not =0$. Therefore, the proof is analogous to the above one to show that $\Phi(\mathbf{v})\not\in K(H)$ with
$\mathbf{v}\in\mathcal{M}$. Then, the result holds.
 
\end{proof}

\begin{theorem}\label{Teo:kernel1}
	Let $\mathcal{H}=\mathcal{H}^{t_1,\dots,t_s}$ be the $\Z_{2^s}$-additive Hadamard code of type $(n; t_1,$ $\dots, t_s)$ such that $\Phi(\mathcal{H})$ is nonlinear. Let $\mathcal{H}_b$ be the subcode of $\mathcal{H}$ which contains all the codewords of order two. Let $P=\lbrace\mathbf{2^p}\rbrace_{p=0}^{\sigma-2}$ if $\sigma\geq2$, and $P=\emptyset$ if $\sigma=1$. Then,
	$$
	\left\langle\Phi(\mathcal{H}_b),\Phi(P), \Phi(\sum_{i=0}^{s-2}\mathbf{2^i})\right\rangle= K(\Phi(\mathcal{H}))
	$$
	and $\kernel(\Phi(\mathcal{H}))=\sigma+\sum_{i=1}^{s}t_i$.
\end{theorem}

\begin{proof}
The result follows by Proposition \ref{Prop:kernel1} and Lemma \ref{lema:nohomoker}.
 
\end{proof}

\begin{corollary}\label{coro:kernelBasis}
	Let $\mathcal{H}=\mathcal{H}^{t_1,\dots,t_s}$ be the $\Z_{2^s}$-additive Hadamard code of type $(n; t_1,$ $\dots, t_s)$ such that $\Phi(\mathcal{H})$ is nonlinear. Let $\mathbf{w}_i$ be the $i$th row of $A^{t_1,\dots,t_s}$ and $\tau=\sum_{i=1}^st_i$. Let $Q=\lbrace (\ord(\mathbf{w}_q)/2)\mathbf{w}_q\rbrace_{q=0}^{\tau}$ and $P=\lbrace\mathbf{2^p}\rbrace_{p=0}^{\sigma-2}$ if $\sigma\geq2$, and $P=\emptyset$ if $\sigma=1$. Then,
	$\{ \Phi(Q), \Phi(P), \Phi(\sum_{i=0}^{s-2}\mathbf{2^i} \}$ is a basis of $K(\Phi(\mathcal{H}))$.
\end{corollary}

\begin{example}\label{Example:kerH200}
Let $H^{2,0,0}$ be the $\Z_8$-linear Hadamard code considered in Example
\ref{ex:Z8Hadamard}. By Theorem \ref{Teo:kernel1}, we have that
$\kernel(H^{2,0,0})=3$. Moreover, we can construct $K(H^{2,0,0})$ from a basis, by Corollary \ref{coro:kernelBasis}. First, we have that $Q=\{ \mathbf{4},(0,4,0,4,0,4,0,4)\}$. Since $\sigma=1$, in this case, we have that $P=\emptyset$. Thus,
$$
K(H^{2,0,0})=\langle\Phi(\mathbf{4}),\Phi((0,4,0,4,0,4,0,4)),\Phi(\mathbf{3})\rangle.
$$
\end{example}


\section{Classification of $\Z_{2^s}$-Linear Hadamard Codes}
\label{classification}

The classification of the $\Z_4$-linear Hadamard codes of length $2^t$, for any $t\geq 3$,
using the rank or the dimension of the kernel is shown in \cite{Kro:2001:Z4_Had_Perf,PheRifVil:2006}.
In this section, we show that the dimension of the kernel can not be used to establish a complete classification of the $\Z_{2^s}$-linear Hadamard codes of length $2^t$, in general, for any $t\geq 3$ and $s>2$. However, we see that this invariant allows us to show some partial results on the classification of these codes, through some examples.

First of all, recall that, for any $t\geq 3$, only the $\Z_4$-linear Hadamard codes $H^{1,t_2}$ and $H^{2,t_2}$ of length $2^t$ are linear \cite{Kro:2001:Z4_Had_Perf},  so these are equivalent to the Reed-Muller code $RM(1,t)$.
By Theorem \ref{theorem:lineal}, for any $t\geq 3$ and $s>2$, there are also at most two $\Z_{2^s}$-linear Hadamard codes of length $2^t$, $H^{1,0,\dots,0,1,t_s}$ and $H^{1,0,\dots,0,t_s}$, that are linear. Moreover, the following result implies that we can focus on $t\geq 5$ and $2\leq s \leq t-2$ to try to classify the nonlinear ones.

\begin{theorem} \label{teo:linearCases}
Let $\cA_{t,s}$ be the number of nonequivalent $\Z_{2^s}$-linear Hadamard codes of length $2^t$. Then,
$$
\cA_{t,s}=\left\lbrace\begin{array}{ll}
0 & \textrm{if } t\geq 3 \textrm{ and } s \geq t+2\\
1 & \textrm{if } t\geq 3 \textrm{ and } s \in \{t-1,t,t+1\}\\
1 & \textrm{if } t=4 \textrm{ and } s=2,\\
\end{array}\right.
$$
and the $\Z_{2^s}$-linear Hadamard code is linear when $\cA_{t,s}=1$.
Moreover, if $t\geq 5$ and $2\leq s \leq t-2$, then $\cA_{t,s} \geq 2$, and there is one which is linear and at least one which is nonlinear.
\end{theorem}

\begin{proof}
First, if $t\geq 3$ and $s \geq t+2$, then the equation
\begin{equation}\label{eq:Equation}
t=\Big(\sum_{i=1}^{s}(s-i+1)\cdot t_i\Big)-1,
\end{equation}
with $t_1\geq 1$, has not any nonnegative integer solution, so $\cA_{t,s}=0$.
If $t\geq 3$ and $s=t+1$, then (\ref{eq:Equation}) has only one solution $(t_1,\ldots,t_s)=(1,0,\ldots,0)$.
If $t\geq 3$ and $s=t$, (\ref{eq:Equation}) has only the solution $(1,0,\ldots,0,1)$.
If $t\geq 3$ and $s=t-1$, (\ref{eq:Equation}) has exactly two solutions $(1,0,\ldots,0,2)$ and
(1,0,\ldots,0,1,0). Note that, when $t=3$ and $s=2$, both solutions are $(1,2)$ and $(2,0)$.
By Theorem \ref{theorem:lineal}, for all the above solutions, we obtain a linear code $H^{t_1,\ldots,t_s}$.

Finally, if $t\geq 5$ and $2\leq s \leq t-2$, (\ref{eq:Equation}) always has the solutions $(1,0,\ldots,0,t-s+1)$ and $(1,0,\ldots,0,1,t-s-1)$, which give a linear code. However, for these cases, there is at least another solution.
On the one hand, if $s=2$, $\cA_{t,s}=\lfloor (t-1)/2\rfloor\geq 2$ since $t\geq 5$ \cite{Kro:2001:Z4_Had_Perf}.
On the other hand, if $s=3$, $(2,0,\cdots,0,t-2s+1)$ is a solution since $t\geq 2s-1$ when $t\geq 5$;
and if $s\geq 4$, $(1,0,\cdots,0,1,0,t-s-2)$ is a solution. Therefore, for all the cases, $\cA_{t,s}\geq 2$ by Theorem \ref{theorem:lineal}.
 
\end{proof}

\medskip
The following example shows that the dimension of the kernel can not be used, in general,
to classify completely all nonlinear $\Z_{2^s}$-linear Hadamard codes of length $2^t$, once $t\geq 5$ and $2 < s \leq t-2$ are fixed.

\begin{example}\label{Example:NoKernel}
The $\Z_8$-linear Hadamard codes of length $2^t=256$ are the following: $H^{1,0,6},H^{1,1,4},H^{1,2,2},H^{1,3,0},H^{2,0,3},$ $H^{2,1,1}$ and $H^{3,0,0}$. The first two are equivalent as they are linear by Theorem \ref{theorem:lineal}. The remaining ones have kernels of dimension $7, 6, 6, 5$ and $4$, respectively, by Theorem \ref{Teo:kernel1}. Therefore, by using this invariant, we can say that all of them are nonequivalent, with the exception of $H^{1,3,0}$ and $H^{2,0,3}$ which have the same dimension of the kernel. For these two codes, by using the computer algebra system Magma \cite{Magma}, we have computed that $\rank(H^{1,3,0})=12$ and $\rank(H^{2,0,3})=11$, so they are also nonequivalent. Actually, all these nonlinear codes have ranks $10, 12, 11, 13$ and $17$, respectively, so we can use the rank instead of the dimension of the kernel to classify completely the $\Z_8$-linear Hadamard codes of length $256$.
\end{example}

As shown in the next example, for some values of $t\geq 5$ and $2< s \leq t-2$, it is indeed possible to establish a complete classification by using just the dimension of the kernel, like it happens for any $t\geq 5$ and $s=2$ \cite{Kro:2001:Z4_Had_Perf}.

\begin{example}
\label{Example:NoKernel2}
By Theorem \ref{Teo:kernel1}, it is possible to check that for any $5\leq t \leq 7$ and $2\leq s\leq t-2$, all nonlinear $\Z_{2^s}$-linear Hadamard codes of length $2^t$ have a different dimension of the kernel, so this invariant allows us to classify them. For $t=8$, $t=9$, $t=10$ and $t=11$, it also works, except when $s\in \{3\}$, $s\in \{3,4\}$, $s\in \{3,4,5\}$ and $s\in \{3,4,5,6\}$, respectively.
For these given values of $t$ and $s$, we can just obtain a partial classification by using the kernel.
\end{example}

\begin{table}[h!]
\centering
\begin{tabular}{|c||cc|cc|cc|}
\cline{1-7}
& \multicolumn{2}{c|}{$t=5$}& \multicolumn{2}{c|}{$t=6$}& \multicolumn{2}{c|}{$t=7$}\\
\cline{2-7}
& $(t_1,\ldots,t_s)$ & $(r,k)$ & $(t_1,\ldots,t_s)$ & $(r,k)$ & $(t_1,\ldots,t_s)$ & $(r,k)$ \\[0.2em]
\hline
\multirow{2}{*}{$\Z_4$} & $(3,0)$  & (7,4)& $(3,1)$  & (8,5)  & $(3,2)$  & (9,6)  \\
                        &          &      &          &        & $(4,0)$  & (11,5) \\[0.2em]
\hline
\multirow{3}{*}{$\Z_8$} & $(2,0,0)$ & (8,3) & $(1,2,0)$ & (8,5)  & $(1,2,1)$ & (9,6)  \\
                        &           &       & $(2,0,1)$ & (9,4)  & $(2,0,2)$ & (10,5) \\
                        &           &       &           &        & $(2,1,0)$ & (12,4) \\[0.2em]
\hline
\multirow{3}{*}{$\Z_{16}$}                           &             &       & $(1,1,0,0)$ & (9,4)  & $(1,0,2,0)$ & (9,6) \\
                           &             &       &             &        & $(1,1,0,1)$ & (10,5)\\
                           &             &       &             &        & $(2,0,0,0)$ & (14,3)\\[0.2em]
\hline
\multirow{1}{*}{$\Z_{32}$}                           &               &       &               &        & $(1,0,1,0,0)$ &  (10,5)\\[0.2em]
\hline
\end{tabular}
\caption{Rank and kernel for all nonlinear $\Z_{2^s}$-linear Hadamard codes of length $2^t$.}
\label{table:Types}
\end{table}

By using Magma, we have also computed the rank of the nonlinear $\Z_{2^s}$-linear Hadamard codes of length $2^t$, for any $5\leq t\leq 11$ and  $2 \leq s \leq t-2$. Tables \ref{table:Types} and \ref{table:TypesB} show the values of $(t_1,\ldots,t_s)$ and the pair $(r,k)$, where $r$ is the rank and $k$ the dimension of the kernel, for all nonlinear $\Z_{2^s}$-linear Hadamard codes of length $2^t$, for $5\leq t\leq 10$.
Note that the results given by Example \ref{Example:NoKernel}, and Example \ref{Example:NoKernel2} for $5\leq t\leq 10$, can also be checked by looking at these tables.
These tables also show that all nonlinear $\Z_{2^s}$-linear Hadamard codes of length $2^t$ have a different value of the rank, once $5\leq t\leq 10$ and  $2 \leq s \leq t-2$ are fixed. Therefore, for these cases, as in Example \ref{Example:NoKernel}, we have that the codes are pairwise nonequivalent, so we have a complete classification by using the rank and we can establish the following result.

\begin{theorem} \label{teo:numNonEquiv}
Let $\cA_{t,s}$ be the number of nonequivalent $\Z_{2^s}$-linear Hadamard codes of length $2^t$. Then,
for any $t \geq 3$ and $2 \leq s \leq t-1$,
$$\cA_{t,s} \leq |\{ (t_1,\ldots,t_s)\in \N^s :  t=\Big(\sum_{i=1}^{s}(s-i+1)\cdot t_i \Big)-1, \ t_1\geq 1 \}|-1.$$
Moreover, for any $3 \leq t \leq 11$ and $2 \leq s \leq t-1$, this bound is tight.
\end{theorem}

\begin{proof}
Straightforward from Theorem \ref{theorem:lineal},
the proof of Theorem \ref{teo:linearCases}, Tables \ref{table:Types} and \ref{table:TypesB}, and further computations by using Magma for $t=11$.
 
\end{proof}

\medskip
By Theorems \ref{teo:numNonEquiv} and \ref{teo:linearCases} (or Tables \ref{table:Types} and \ref{table:TypesB}), we can obtain exactly the number of nonequivalent $\Z_{2^s}$-linear Hadamard codes of length $2^t$, for some values of $t$ and $s$. Table \ref{table:NumberNonequivalent} shows these numbers, for $3 \leq t \leq 11$ and $2 \leq s \leq 9$. The cases where the dimension of the kernel is not enough to classify these codes are shown in bold type.
However, in all these cases, the rank can be used to obtain the classification.

\begin{table}[h]
\centering
\begin{tabular*}{0.5\textwidth}{@{\extracolsep{\fill}}|c||c|c|c|c|c|c|c|c|c|}
\cline{0-9}
$t$ & 3 & 4 & 5 & 6 & 7 & 8 & 9 & 10 & 11\\\cline{0-9}
$\Z_4$ & 1 & 1 & 2 & 2 & 3 & 3 & 4 & 4 & 5\\\cline{0-9}
$\Z_8$ & 1 & 1 & 2 & 3 & 4 & {\bf 6} & {\bf 7} & {\bf 9} & {\bf 11}\\\cline{0-9}
$\Z_{16}$ & 1 & 1 & 1 & 2 & 4 & 5 & {\bf 8} & {\bf 10} & {\bf 14}\\\cline{0-9}
$\Z_{32}$ & 0 & 1 & 1 & 1 & 2 & 4 & 6 & {\bf 9} & {\bf 12}  \\\cline{0-9}
$\Z_{64}$ & 0 & 0 & 1 & 1 & 1 & 2 & 4 & 6 & {\bf 10}  \\\cline{0-9}
$\Z_{128}$& 0 & 0 & 0 & 1 & 1 & 1 & 2 & 4 & 6  \\\cline{0-9}
$\Z_{256}$ & 0 & 0 & 0 & 0 & 1 & 1 & 1 & 2 & 4   \\\cline{0-9}
$\Z_{512}$& 0 & 0 & 0 & 0 & 0 & 1 & 1 & 1 & 2   \\\cline{0-9}
\end{tabular*}
\caption{Number $\cA_{t,s}$ of nonequivalent $\Z_{2^s}$-linear Hadamard codes of length $2^t$.}
\label{table:NumberNonequivalent}
\end{table}

The values of $\cA_{t,2}$ given in Table \ref{table:NumberNonequivalent} where already proved in \cite{Kro:2001:Z4_Had_Perf}.
Specifically, in that paper, it is shown that there are $\lfloor
\frac{t-1}{2}\rfloor$ nonequivalent $\Z_4$-linear Hadamard  codes of length $2^t$ for all $t\geq 3$.
Next, we focus on establishing some relationships between the already known $\Z_{2^s}$-linear Hadamard codes with $s=2$ and the ones with $s>2$, once only the length $2^t$ is fixed. First, Example \ref{ex:CompZ4A} shows that there are $\Z_{2^s}$-linear Hadamard codes, with $s>2$, which are not equivalent to any $\Z_4$-linear Hadamard code. Then, Example \ref{ex:CompZ4B} also shows that there are $\Z_4$-linear Hadamard codes which are not equivalent to any $\Z_{2^s}$-linear Hadamard codes with $s>2$.

\begin{example} \label{ex:CompZ4A}
Let $H^{2,0,0}$ be the $\Z_8$-linear Hadamard code of length $32$ considered in Examples
\ref{ex:Z8Hadamard} and \ref{Example:kerH200}. Recall that $\kernel(H^{2,0,0})=3$ by Theorem \ref{Teo:kernel1}, and hence $H^{2,0,0}$ is nonlinear.
It is known that there are three  $\Z_4$-linear Hadamard codes of length $32$, $H^{1,4}$, $H^{2,2}$ and $H^{3,0}$.
The first two are linear, and the last one has $\kernel(H^{3,0})=4$ by Theorem \ref{Teo:kernel1} or \cite{Kro:2001:Z4_Had_Perf}. Hence,
there is no $\Z_4$-linear Hadamard code equivalent to the $\Z_8$-linear Hadamard code $H^{2,0,0}$.
\end{example}

\begin{example} \label{ex:CompZ4B}
By Table \ref{table:Types}, for $t=5$, there are only two nonlinear $\Z_{2^s}$-linear Hadamard codes, $H^{3,0}$ and $H^{2,0,0}$.
In Example \ref{ex:CompZ4A}, we have seen that they are not equivalent, since they have different dimension of the kernel.
Other examples like this one can be found when $t$ is odd. For example, by Tables \ref{table:Types} and \ref{table:TypesB} and further computations in Magma, for $t=7$, $t=9$ and $t=11$,
there are $\Z_4$-linear Hadamard codes, $H^{4,0}$, $H^{5,0}$ and $H^{6,0}$, respectively,
which are not equivalent to any $\Z_{2^s}$-linear Hadamard codes  with $s>2$ of the same length,
by using both invariants, the rank and the dimension of the kernel.
\end{example}

It is also worth to mention that there are Hadamard codes, called $\Z_2\Z_4$-linear Hadamard codes,
which came from the image of a generalized Gray map of subgroups of $\Z_2^\alpha \times \Z_4^\beta$.
Note that if $\alpha=0$, they correspond to $\Z_4$-linear Hadamard codes. More information on $\Z_2\Z_4$-linear codes in general can be found in \cite{ccsg}. The classification of $\Z_2\Z_4$-linear Hadamard codes of length $2^t$ with $\alpha\not =0$ is given in \cite{PheRifVil:2006}, where
it is shown that there are $\lfloor \frac{t}{2} \rfloor$ nonequivalent of such codes, for all $t\geq 3$;
and either the rank or the dimension of the kernel
can be used to classify them, like for $\Z_4$-linear Hadamard codes. Recall that there are $\lfloor
\frac{t-1}{2}\rfloor$ nonequivalent $\Z_4$-linear Hadamard  codes of length $2^t$ for all $t\geq 3$ \cite{Kro:2001:Z4_Had_Perf}.
However, in \cite{KroVil2015}, it is shown that each $\Z_2\Z_4$-linear Hadamard
code with $\alpha=0$, that is, each $\Z_4$-linear Hadamard code, is equivalent to a $\Z_2 \Z_4$-linear Hadamard  code with $\alpha \not =0$,
so there are only $\lfloor\frac{t}{2} \rfloor$ nonequivalent $\Z_2\Z_4$-linear Hadamard codes of length $2^t$.

The following example shows that there are $\Z_2\Z_4$-linear Hadamard codes (with $\alpha \not =0$)
which are not equivalent to any $\Z_{2^s}$-linear Hadamard codes with $s\geq 2$.

\begin{example}\label{ex:CompZ2Z4}
For $t=4$, there is a $\Z_2\Z_4$-linear Hadamard code (with $\alpha \not =0$) which is not equivalent to any $\Z_4$-linear Hadamard code \cite{KroVil2015}. This code has parameters $(r,k)=(6,3)$ \cite{PheRifVil:2006}, so it is not equivalent to any $\Z_{2^s}$-linear Hadamard code with $s\geq 2$, since all of them are linear by Theorem \ref{teo:linearCases}. Other examples like this one can be found when $t$ is even. For example, for $t=6$, $t=8$ and $t=10$, there is also a $\Z_2\Z_4$-linear Hadamard code (with $\alpha \not =0$) which is not equivalent to any $\Z_4$-linear Hadamard code \cite{KroVil2015}. They have parameters $(10,4)$, $(15,5)$ and $(21,6)$ \cite{PheRifVil:2006}, respectively, so again they are not equivalent to any $\Z_{2^s}$-linear Hadamard code with $s\geq 2$ of length $2^6$, $2^8$ and $2^{10}$, respectively, by Tables \ref{table:Types} and \ref{table:TypesB}. 
\end{example}


Finally, we focus on establishing how many nonequivalent $\Z_{2^s}$-linear Hada\-mard codes of length $2^t$ there are, once only the
length $2^t$ is fixed for some values of $t$. First, we give some lower and upper bounds. From Tables \ref{table:Types} and \ref{table:TypesB}, and further computations in Magma for $t=11$, we can determine a lower bound (K) taking into account just the dimension of the kernel. This lower bound can be improved (RK) if we consider both invariants, the rank and the dimension of the kernel.  These results are summarized in Table \ref{table:NumNonEquivalentCodes}, where we give these bounds for all $3\leq t \leq 11$.

\begin{table}[h]
\small
\centering
\begin{tabular}{|c||c|c|c|c|c|c|c|c|c|}
\cline{0-9}
$t$ & 3 & 4 & 5 & 6 & 7 & 8 & 9 & 10 & 11\\ \cline{0-9} \cline{0-9}
lower bound K & 1  & 1  & 3  & 3 & 5 & 5 & 7 & 7 &  9     \\\cline{0-9}
lower bound RK & 1  & 1  & 3  & 3 & 6 & 7 & 11 & 13 & 20 \\\cline{0-9} \cline{0-9}
upper bound   & 1  & 1  & 3  & 5 & 10 & 16 & 26 & 38 & 57   \\\cline{0-9}
\end{tabular}
\caption{Bounds for the number $\cA_{t}$ of nonequivalent $\Z_{2^s}$-linear Hadamard codes of length $2^t$.}
\label{table:NumNonEquivalentCodes}
\end{table}

An upper bound can be given easily by considering all nonequivalent $\Z_{2^s}$-linear Hadamard codes of length $2^t$, once $t$ and $s$ are fixed,
as it is shown in the next theorem. These values for all $3\leq t \leq 11$ are also shown in Table \ref{table:NumNonEquivalentCodes}.

\begin{theorem}
Let $\cA_{t,s}$ be the number of nonequivalent $\Z_{2^s}$-linear Hadamard codes of length $2^t$.
Let $\cA_{t}$ be the number of nonequivalent $\Z_{2^s}$-linear Hadamard codes of length $2^t$, for any $s\geq 2$.
Then, $\cA_{t} \leq \sum_{s=2}^{t-2} (\cA_{t,s}-1) +1$.
\end{theorem}

\begin{theorem}
There are exactly 1,1,3, 3 and 6 nonequivalent $\Z_{2^s}$-linear Hada\-mard codes of length $2^t$ for $t$ equal to $3,4,5,6$ and $7$, respectively.
\end{theorem}

\begin{proof}
For $t$ equal to 3, 4 and 5, the result is true, since the lower and upper bounds given in Table \ref{table:NumNonEquivalentCodes} coincides.
By using Magma, it is possible to check that, for $t=6$, both $\Z_{2^s}$-linear Hadamard codes having the same parameters $(r,k)=(8,5)$
are equivalent; and the ones having $(r,k)=(9,4)$ are also equivalent. Therefore, in this case, the upper bound goes from 5 to 3, and
then coincides with the lower bound given in Table \ref{table:NumNonEquivalentCodes}. Similarly, or $t=7$, it is also possible to check
that the codes having the same parameters $(r,k)$ are all equivalent, so the upper bound became equal to the lower bound 6, and the result also holds.
 
\end{proof}

\begin{table}[h!]
\begin{tabular}{|c||cc|cc|cc|}
\cline{1-7}
&\multicolumn{2}{c|}{$t=8$}& \multicolumn{2}{c|}{$t=9$}& \multicolumn{2}{c|}{$t=10$}\\
\cline{2-7}
& $(t_1,\ldots,t_s)$ & $(r,k)$ & $(t_1,\ldots,t_s)$ & $(r,k)$ & $(t_1,\ldots,t_s)$ & $(r,k)$ \\[0.2em]
\hline
\multirow{3}{*}{$\Z_4$}                        & $(3,3)$  & (10,7) & $(3,4)$ & (11,8)  & $(3,5)$ & (12,9)  \\
                        & $(4,1)$  & (12,6) & $(4,2)$ & (13,7)  & $(4,3)$ & (14,8)  \\
                        &         &         & $(5,0)$ & (16,6)  & $(5,1)$ & (17,7)  \\[0.2em]
\hline
\multirow{6}{*}{$\Z_8$}  & $(1,2,2)$ & (10,7) & $(1,2,3)$ & (11,8) & $(1,2,4)$ & (12,9)  \\
                        & $(1,3,0)$ & (12,6) & $(1,3,1)$ & (13,7) & $(1,3,2)$ & (14,8)  \\
                        & $(2,0,3)$ & (11,6) & $(2,0,4)$ & (12,7) & $(1,4,0)$ & (17,7)  \\
                        & $(2,1,1)$ & (13,5) & $(2,1,2)$ & (14,6) & $(2,0,5)$ & (13,8)  \\
                        & $(3,0,0)$ & (17,4) & $(2,2,0)$ & (17,5) & $(2,1,3)$ & (15,7)  \\
                        &           &        & $(3,0,1)$ & (18,5) & $(2,2,1)$ & (18,6)  \\
                        &           &        &           &        & $(3,0,2)$ & (19,6)  \\
                        &           &        &           &        & $(3,1,0)$ & (24,5)  \\[0.2em]
\hline
\multirow{7}{*}{$\Z_{16}$}  & $(1,0,2,1)$ & (10,7) & $(1,0,2,2)$ & (11,8) & $(1,0,2,3)$ & (12,9)  \\
                           & $(1,1,0,2)$ & (11,6) & $(1,0,3,0)$ & (13,7) & $(1,0,3,1)$ & (14,8)  \\
                           & $(1,1,1,0)$ & (13,5) & $(1,2,0,0)$ & (18,5) & $(1,1,0,4)$ & (13,8)  \\
                           & $(2,0,0,1)$ & (15,4) & $(1,1,0,3)$ & (12,7) & $(1,1,1,2)$ & (15,7)  \\
                           &             &        & $(1,1,1,1)$ & (14,6) & $(1,1,2,0)$ & (18,6)  \\
                           &             &        & $(2,0,0,2)$ & (16,5) & $(1,2,0,1)$ & (19,6)  \\
                           &             &        & $(2,0,1,0)$ & (20,4) & $(2,0,0,3)$ & (17,6)  \\
                           &             &        &             &        & $(2,0,1,1)$ & (21,5)  \\
                           &             &        &             &        & $(2,1,0,0)$ & (28,4)  \\[0.2em]
\hline
\multirow{5}{*}{$\Z_{32}$}  & $(1,0,0,2,0)$ & (10,7) & $(1,0,0,2,1)$ & (11,8) & $(1,0,0,2,2)$ & (12,9)  \\
                           & $(1,0,1,0,1)$ & (11,6) & $(1,0,1,0,2)$ & (12,7) & $(1,0,0,3,0)$ & (14,8)  \\
                           & $(1,1,0,0,0)$ & (15,4) & $(1,0,1,1,0)$ & (14,6) & $(1,0,1,0,3)$ & (13,8)  \\
                           &               &        & $(1,1,0,0,1)$ & (16,5) & $(1,0,1,1,1)$ & (15,7)  \\
                           &               &        & $(2,0,0,0,0)$ & (26,3) & $(1,0,2,0,0)$ & (19,6)  \\
                           &               &        &               &        & $(1,1,0,0,2)$ & (17,6)  \\
                           &               &        &               &        & $(1,1,0,1,0)$ & (21,5)  \\
                           &               &        &               &        & $(2,0,0,0,1)$ & (27,4)  \\[0.2em]
\hline
\multirow{3}{*}{$\Z_{64}$} & $(1,0,0,1,0,0)$ & (11,6)& $(1,0,0,0,2,0)$ & (11,8)  & $(1,0,0,0,2,1)$ & (12,9) \\
                           &                 &       & $(1,0,0,1,0,1)$ & (12,7)  & $(1,0,0,1,0,2)$ & (13,8) \\
                           &                 &       & $(1,0,1,0,0,0)$ & (16,5)  & $(1,0,0,1,1,0)$ & (15,7) \\
                           &                 &       &                 &         & $(1,0,1,0,0,1)$ & (17,6) \\
                           &                 &       &                 &         & $(1,1,0,0,0,0)$ & (27,4) \\[0.2em]
\hline
\multirow{3}{*}{$\Z_{128}$}&                 &       & $(1,0,0,0,1,0,0)$ & (12,7)  & $(1,0,0,0,0,2,0)$ & (12,9) \\
                           &                 &       &                   &         & $(1,0,0,0,1,0,1)$ & (13,8) \\
                           &                 &       &                   &         & $(1,0,0,1,0,0,0)$ & (17,6) \\[0.2em]
\hline
\multirow{1}{*}{$\Z_{256}$}                 &                 &       &                 &         & $(1,0,0,0,0,1,0,0)$ & (13,8) \\[0.2em]
\hline
\end{tabular}
\caption{Rank and kernel for all nonlinear $\Z_{2^s}$-linear Hadamard codes of length $2^t$.}
\label{table:TypesB}
\end{table}


\section{Conclusions}
\label{Sec:Conclusions}

The kernel of $\Z_{2^s}$-linear Hadamard codes of length $2^t$ has been studied for $s>2$.
We compute the kernel of these codes and its dimension in order to classify them, like it is done for $s=2$.
We first have considered that the parameters $t$ and $s$ are fixed.
Examples \ref{Example:NoKernel} and \ref{Example:NoKernel2} show that, only for some values of $t$ and $s$,
we can use the dimension of the kernel to distinguish between nonequivalent
$\Z_{2^s}$-linear Hadamard codes of length $2^t$.
However, when this is not possible, by using Magma, we also show that in these examples the rank is enough to classify them.
A further research on this topic would be to determine exactly for which values of $t$ and $s$ the
dimension of the kernel classify them, or to compute the rank of these codes in order to prove that
this invariant is enough to established their complete classification, for any $t$ and $s$.

We have also considered that only the length $2^t$ is fixed. Examples \ref{ex:CompZ4B} and \ref{ex:CompZ2Z4} seem to point out that
for any $t$ odd ($t$ even), there is exactly one $\Z_4$-linear Hadamard code ($\Z_2\Z_4$-linear Hadamard code with $\alpha\not =0$)
of length $2^t$  which is not equivalent to any $\Z_{2^s}$-linear
Hadamard code with $s>2$; and the others $\Z_4$-linear Hadamard codes ($\Z_2\Z_4$-linear Hadamard codes with $\alpha\not =0$) of the same length are equivalent to a $\Z_8$-linear Hadamard code,
although they may also have other structures with $s>3$.
Another further research in this sense would be to establish whether this is always true; and whether
all $\Z_{2^s}$-linear Hadamard codes of the same length, having the same rank and dimension of the kernel, are equivalent.

Finally, we have established some bounds for the exact number of nonequivalent $\Z_{2^s}$-linear Hadamard codes of
length $2^t$, when both $t$ and $s$ are fixed, and when just $t$ is fixed; denoted by $\cA_{t,s}$ and $\cA_{t}$, respectively.
For some values of $t$ and $s$, we provide their exact values, by using Magma.
It would also be interesting to determine them for any possible $t$ and $s$, as a further research; as well as equivalent results by
considering the generalized Gray map used in \cite{Krotov:2007}.


\begin{thebibliography}{10}\label{bibliography}

\bibitem{TwoWeightSole} Alahmadi, A., Sepasdar, Z., Shi, M., Sol\'e, P.: On two-weight $\Z_{2^k}$-codes.
Designs, Codes and Cryptography (2017). doi:10.1007/s10623-017-0390-0

\bibitem{Key} Assmus, E. F., Key, J. D.: Designs and
their codes. Cambridge University Press, Great Britain (1992).

\bibitem{BGH83} Bauer, H., Ganter, B., Hergert, F.: Algebraic techniques for
nonlinear codes. Combinatorica, 3, no. 1, pp. 21--33 (1983).

\bibitem{Blake} Blake, I.F.: Codes over integer residue rings. Inform. Control, 29, pp. 295--300 (1975).

\bibitem{ccsg} Borges, J., Fern\'{a}ndez-C\'{o}rdoba, C., Pujol, J., Rif\`a, J.,
Villanueva, M.: $\Z_2\Z_4$-linear codes: generator matrices and duality.
Designs, Codes and Cryptography, 54, no. 2, pp. 167--179 (2010).

\bibitem{Magma} Bosma W., Cannon, J. J., Fieker, C., Steel, A.: Handbook of
Magma functions, Edition 2.22 5669 pages (2016).
{\texttt{http://magma.maths.usyd.edu.au/magma/}}.

\bibitem{Carlet} Carlet, C.: $\Z_{2^k}$-linear codes. IEEE Trans. Inform.
Theory, 44, no. 4, pp. 1543--1547 (1998).

\bibitem{dougherty} Dougherty, S. T., Fern\'{a}ndez-C\'{o}rdoba, C.: Codes Over
$\Z_{2^ k}$, Gray Map and Self-Dual Codes. Advances in Mathematics of
Communications, 5, no. 4, pp. 571--588 (2011).

\bibitem{5ICMCTA} Fern\'andez-C\'ordoba, C., Vela, C., Villanueva, M.: On the Kernel of $\Z_{2^s}$-linear Hadamard Codes.
Coding theory and Applications, ICMCTA 2017. Lecture Notes in Computer Science, 10495, pp. 107--117 (2017).

\bibitem{JMDA} Fern\'andez-C\'ordoba, C., Vela, C., Villanueva, M.: Construction and Classification of the
$\Z_{2^s}$-linear Hadamard codes. Electronic Notes in Discrete Mathematics, 54, pp. 247--252 (2016).

\bibitem{Gupta-paper} Gupta, M. K., Bhandari M.C., Lal, A.K.: On some linear codes over $\Z_{2^s}$. Designs, Codes and Cryptography, 36, no. 3, pp. 227--244 (2005).

\bibitem{Sole} Hammons, A. R., Kumar, P. V., Calderbank, A. R., Sloane, N. J.
A., Sol\'{e}, P.: The $\Z_4$-linearity of Kerdock, Preparata, Goethals and
related codes. IEEE Trans. Inform. Theory, 40, no. 2, pp. 301--319 (1994).

\bibitem{Nechaev} Honold, T., Nechaev, A. A.: Weighted modules and representations of codes. Probl. Inf. Transm., 35, no. 3, pp. 18–39 (1999).

\bibitem{Krotov:2007} Krotov, D. S.: On $\Z_{2^k}$-dual binary codes. IEEE
Trans. Inf. Theory, 53, no 4, pp. 1532--1537 (2007).

\bibitem{Kro:2001:Z4_Had_Perf} Krotov, D. S.: {$\Z_4$}-linear {H}adamard and
extended perfect codes. International Workshop on Coding and Cryptography, ser.
{Electron. Notes Discrete Math.} 6, pp. 107--112 (2001).

\bibitem{KroVil2015} Krotov, D. S., Villanueva, M.: Classification of the
$\Z_2\Z_4$-linear Hadamard codes and their automorphism groups. IEEE Trans. Inf.
Theory, 61, no. 2, pp. 887--894 (2015).

\bibitem{WMcwill} MacWilliams, F. J., Sloane, N. J. A.: The theory of
error-correcting codes. 16, Elsevier (1977).

\bibitem{PheRifVil:2006} Phelps, K. T., Rif\`a, J., Villanueva, M.: On the
additive ({$\Z_4$}-linear and non-{$\Z_4$}-linear) {H}adamard codes: rank and
kernel. IEEE Trans. Inf. Theory, 52, no. 1, pp. 316--319 (2006).

\bibitem{Shankar} Shankar, P.: On BCH codes over arbitrary integer rings.  IEEE Trans. Inf. Theory, 25, no. 4, pp. 480--483 (1979).

\bibitem{TapVeg:2003} Tapia-Recillas, H., Vega, G.: On $\Z_{2^k}$-linear and
quaternary codes. SIAM J. Discrete Math., 17, no. 1, pp. 103--113 (2003).

\end{thebibliography}
\end{document}